\documentclass[conference]{IEEEtran}
 
\usepackage[utf8]{inputenc}

\usepackage{booktabs, tabularx, makecell}

\let\C\relax

\usepackage{algorithm}
\usepackage[noend]{algpseudocode}
\usepackage{newfile}
\usepackage{complexity}
\usepackage{todonotes}
\usepackage{bm}
\usepackage{mathtools}
\usepackage{tikz,tikz-qtree}
\usetikzlibrary{automata,decorations.pathreplacing}
\usepackage{thm-restate}
\usepackage{shuffle}
\usepackage{wrapfig}

\usepackage{amsfonts}
\usepackage{amsmath}
\usepackage{amsthm}
\usepackage{amssymb}

\usepackage[colorlinks=true,citecolor=.,linkcolor=.,hypertexnames=false]{hyperref}
\usepackage[noabbrev,nameinlink,capitalise]{cleveref}

\newtheorem{theorem}{Theorem}
\newtheorem{proposition}{Proposition}
\newtheorem{lemma}{Lemma}

\theoremstyle{definition}

\theoremstyle{definition}
\newtheorem{example}{Example}
\theoremstyle{definition}

\newcommand{\Rec}{\mathbf{Rec}}
\newcommand{\Sync}{\mathbf{Sync}}
\newcommand{\DRat}{\mathbf{DRat}}
\newcommand{\Rat}{\mathbf{Rat}}

\newcommand{\wRec}{\omega\text{-}\mathbf{Rec}}
\newcommand{\wSync}{\omega\text{-}\mathbf{Sync}}
\newcommand{\wDRat}{\omega\text{-}\mathbf{DRat}}
\newcommand{\wRat}{\omega\text{-}\mathbf{Rat}}

\newcommand{\OMIT}[1]{}

\renewcommand{\A}{\mathcal{A}}
\newcommand{\B}{\mathcal{B}}
\renewcommand{\C}{\mathbf{C}}

\newcommand{\I}{\mathcal{I}}

\newcommand{\N}{\mathbb{N}}

\newcommand{\Z}{\mathbb{Z}}
\renewcommand{\R}{\mathbb{R}}

\newcommand{\coREXP}{\mathsf{coREXP}}

\newcommand{\rev}{\mathsf{rev}}

\newcommand{\bsvector}[2]{\begin{bsmallmatrix} #1 \\ #2 \end{bsmallmatrix}}

\newclass{\TwoEXP}{2EXP}

\newcommand{\define}[1]{\stackrel{\textrm{\tiny def}}{#1}}

\newcommand{\threecp}{\xrightarrow{\mathsf{3CP}}}
\newcommand{\threecpf}{\xrightarrow{\mathsf{3CP}}_F}

\newcommand{\ee}{\sim_{\mathsf{e}}}

\makeatletter
\DeclareFontFamily{OMX}{MnSymbolE}{}
\DeclareSymbolFont{MnLargeSymbols}{OMX}{MnSymbolE}{m}{n}
\SetSymbolFont{MnLargeSymbols}{bold}{OMX}{MnSymbolE}{b}{n}
\DeclareFontShape{OMX}{MnSymbolE}{m}{n}{
    <-6>  MnSymbolE5
   <6-7>  MnSymbolE6
   <7-8>  MnSymbolE7
   <8-9>  MnSymbolE8
   <9-10> MnSymbolE9
  <10-12> MnSymbolE10
  <12->   MnSymbolE12
}{}
\DeclareFontShape{OMX}{MnSymbolE}{b}{n}{
    <-6>  MnSymbolE-Bold5
   <6-7>  MnSymbolE-Bold6
   <7-8>  MnSymbolE-Bold7
   <8-9>  MnSymbolE-Bold8
   <9-10> MnSymbolE-Bold9
  <10-12> MnSymbolE-Bold10
  <12->   MnSymbolE-Bold12
}{}

\let\llangle\@undefined
\let\rrangle\@undefined
\DeclareMathDelimiter{\llangle}{\mathopen}%
                     {MnLargeSymbols}{'164}{MnLargeSymbols}{'164}
\DeclareMathDelimiter{\rrangle}{\mathclose}%
                     {MnLargeSymbols}{'171}{MnLargeSymbols}{'171}
\makeatother

\usepackage[capitalise,nameinlink,noabbrev]{cleveref}

\begin{document}
\sloppy 

\title{Revisiting Membership Problems in Subclasses of Rational Relations}

\author{\IEEEauthorblockN{Pascal Bergstr\"{a}{\ss}er}
\IEEEauthorblockA{Department of Computer Science\\
RPTU Kaiserslautern-Landau\\
Kaiserslautern, Germany}
\and
\IEEEauthorblockN{Moses Ganardi}
\IEEEauthorblockA{Max Planck Institute for Software Systems\\ (MPI-SWS)\\
Kaiserslautern, Germany}}

\maketitle

\begin{abstract}
	We revisit the \emph{membership problem} for subclasses of rational relations over finite and infinite words:
	Given a relation $R$ in a class $\C_2$, does $R$ belong to a smaller class $\C_1$?
	The subclasses of rational relations that we consider are formed by the deterministic rational relations,
	synchronous (also called automatic or regular) relations, and recognizable relations.
	For almost all versions of the membership problem, determining the precise complexity or even decidability
	has remained an open problem for almost two decades.
	In this paper, we provide improved complexity and new decidability results.
    (i)~Testing whether a synchronous relation over infinite words is recognizable is
    $\NL$-complete ($\PSPACE$-complete) if the relation is given by a deterministic (nondeterministic) $\omega$-automaton.
    This fully settles the complexity of this recognizability problem,
    matching the complexity of the same problem over finite words.
    (ii)~Testing whether a deterministic rational binary relation is recognizable is
    decidable in polynomial time,
    which improves a previously known double exponential time upper bound.
    For relations of higher arity, we present a randomized exponential time algorithm.
    (iii)~We provide the first algorithm to decide whether a deterministic rational relation is synchronous.
    For binary relations the algorithm even runs in polynomial time.
\end{abstract}


\section{Introduction}

The study of \emph{relations over words} and their computational models, often called \emph{transducers},
has become an active field of research, with applications in various fields,
including algorithmic verification~\cite{AlurC11,ChenHLRW19}, synthesis~\cite{Thomas08,FiliotJLW16},
and graph databases~\cite{BarceloFL13}.
While the class of regular languages is captured by several equivalent automata models,
e.g.\ deterministic and nondeterministic automata,
which read their input either in one or both directions,
the same does not hold anymore for relations.
The literature contains a number of transducer models for relations
with varying tradeoffs between expressivity, closure properties, and algorithmic amenability.
Finite-state transducers reading multiple words have been already introduced by Rabin and Scott
in their seminal paper~\cite{RabinS59}.
This basic model has later been extended to more expressive models
such as streaming string transducers, transductions with origin semantics,
visibly pushdown transducers, and transducers over infinite words,
see~\cite{MuschollP19} for an overview.
Various algorithmic questions on basic transducer models remain challenging open problems,
e.g.\ determining the precise complexity of the equivalence problem for deterministic streaming string transducers~\cite{AlurC11}
or for deterministic multitape automata~\cite{Worrell13}.

\paragraph*{The membership problem}
In this paper we revisit the \emph{membership problem} (or the \emph{definability problem})
for relations over words, i.e.\ given a relation $R$ in a class $\C_2$,
does $R$ belong to a smaller class $\C_1$?
The membership problem for languages is a classical question in automata theory,
in particular the question whether a given regular language belongs to a subclass of the
regular languages~\cite{Schutzenberger65a,Zalcstein72,BrzozowskiS73,Simon75,Place20},
and the question whether a given language from a superclass of the regular languages
is in fact regular~\cite{Greibach68,Valiant75,ValkV81}.
For example, Schützenberger's theorem effectively characterizes which regular languages are star-free~\cite{Schutzenberger65a}.
Deciding whether an NFA accepts a star-free language is \PSPACE-complete~\cite{ChoH91}.
Another milestone result in this context is Valiant's regularity test for deterministic pushdown automata (DPDAs)~\cite{Valiant75}.
Its running time is double exponential, improving on a previous triple exponential time algorithm by Stearns~\cite{Stearns67}.
The only known lower bound is \P-hardness inherited from emptiness problem, leaving an almost fifty year old
double exponential gap between the upper and the lower bound.

The membership problem for relations over words
was first systematically studied by Carton, Choffrut, and Grigorieff~\cite{CartonCG06}
for subclasses of rational relations over finite words.
Let us briefly introduce the most important subclasses, see~\cref{sec:classes} for formal definitions.
A relation is \emph{rational} if it is recognized by a nondeterministic multitape automaton
where the tapes are read asynchronously in one direction~\cite{ElgotM65}.
The deterministic variant of multitape automata~\cite{RabinS59} captures the class of \emph{deterministic rational} relations.
Unfortunately, universality of rational relations and inclusion
of deterministic rational relations are undecidable~\cite{FischerR68}.
To overcome these undecidability barriers, one can put the additional restriction on the automaton
that all heads read their input letter synchronously in parallel.
Synchronous multitape automata recognize the \emph{synchronous (rational) relations}~\cite{FrougnyS93},
also called automatic or regular relations.
Due to their effective closure under first-order operations, they enjoy pleasant algorithmic properties
and form the basis of \emph{automatic structures}~\cite{KhoussainovN94,BlumensathG00}
and of \emph{regular model checking}~\cite{AbdullaJMd02,AbdullaJNS04}.
The smallest class of relations we consider is formed by the \emph{recognizable relations}, where the input words are processed
by independent automata that synchronize only on the sequence of final states reached after reading the entire words.
Alternatively, recognizable relations can be described as finite unions
of Cartesian products of regular languages~\cite[Theorem~1.5]{Berstel79}.
All mentioned classes of relations over finite words are extended to infinite words,
by adding a Büchi condition for nondeterministic automata or a parity condition for deterministic automata.
The hierarchy of the considered subclasses of rational relations over finite and infinite words is displayed in~\cref{fig:results}.

However, as most interesting problems on rational relations,
it is undecidable to test whether a given rational relation
is recognizable, synchronous, or deterministic rational~\cite{FischerR68,Lisovik79}.
Hence, we turn our attention to subclasses of deterministic rational relations.
The following observation from~\cite{CartonCG06} makes a simple connection between \emph{binary} rational relations
$R \subseteq \Sigma^* \times \Sigma^*$ and context-free languages:
If $R$ is rational then $L_R = \{ \rev(u) \# v \mid (u,v) \in R \}$ is context-free where $\rev(u)$ is the reversal of $u$;
if $R$ is deterministic rational then $L_R$ is deterministic context-free.
Furthermore, $R$ is recognizable if and only if $L_R$ is regular.
Therefore, recognizability of binary deterministic rational relations can be easily reduced
to regularity of DPDAs, which can be decided in double exponential time~\cite{Valiant75}.
Using methods from the regularity algorithm, originally due to Stearns~\cite{Stearns67},
one can also decide\footnotemark{} recognizability of deterministic rational relations of arbitrary arity~\cite{CartonCG06}.
\footnotetext{While the authors of~\cite{CartonCG06} did not analyze the complexity of their algorithm,
it is easy to see that their algorithm runs in elementary time for fixed arity $k$.}
Carton, Choffrut, and Grigorieff also present an algorithm to test whether a synchronous relation
is recognizable~\cite{CartonCG06}, which runs in double exponential time (see the remark on its running time in~\cite{LodingS19}).
Recently, Barcel\'o et al.\ determined the precise complexity of the same problem~\cite{BHLLN19}, see below.
The question how to decide whether a deterministic rational relation is synchronous still hitherto remains open.

\paragraph*{Recognizability and the infinite clique problem}
It has been observed in~\cite{CartonCG06} that the recognizability problem for subclasses of rational relations
can be reduced to checking whether certain equivalence relations have finite index.
For a relation $R \subseteq (\Sigma^*)^k$ and $j \in [1,k-1]$ define the equivalence relations $\sim^R_j$
on $(\Sigma^*)^j$ by
\begin{align*}
	\bm{x} \sim_j^R \bm{y} \define \iff & \text{for all } \bm{z} \in (\Sigma^*)^{k-j} \colon \\
	& \qquad \big[ (\bm{x},\bm{z}) \in R \iff (\bm{y},\bm{z}) \in R \big],
\end{align*}
resembling the Myhill-Nerode congruence for languages.
If $R$ is rational, then $R$ is recognizable if and only if $\sim_j^R$ has finite index
for all $j \in [1,k-1]$~\cite[Proposition~3.8]{CartonCG06}.
This characterization has been used in~\cite{BHLLN19} to decide recognizability for synchronous relations, as follows.
Given a DFA (NFA) for a synchronous relation $R$,
one can compute automata for the complement relations $\not \sim_j^R$ in logarithmic space (polynomial space).
Hence, to decide non-recognizability of $R$ it suffices to test
whether for a given \emph{co-equivalence relation} $\not \sim$
there exists an infinite set $X$
such that $\bm{x} \not \sim \bm{y}$ for all distinct $\bm{x}, \bm{y} \in X$, in other words, whether 
$\not \sim$ has an infinite clique.
In fact, the infinite clique problem for arbitrary synchronous relations
was shown to be \NL-complete~\cite{BHLLN19} (later simplified in~\cite{BGLZ22}).

However, in certain settings we need to exploit the fact that $\not \sim_j^R$ is the complement
of an equivalence relation $\sim_j^R$.
For example, Löding and Spinrath~\cite{LodingS19} have shown that the infinite clique problem
for $\omega$-synchronous co-equivalence relations is decidable in double exponential time.
This yields a double (triple) exponential time algorithm for the $\omega$-recognizability problem
for $\omega$-synchronous relations
given by (non)deterministic $\omega$-automata.
Whether the infinite clique problem over \emph{arbitrary} $\omega$-synchronous relations is decidable
is a longstanding open problem~\cite{Kuske10}.
Another example where the difference between co-equivalence relations and arbitrary relations becomes apparent
is the case of \emph{tree-automatic relations}.
It was proven in~\cite{BGLZ22} that the infinite clique problem for tree-automatic relations is \EXP-complete;
however, restricted to complements of transitive relations the infinite clique problem becomes \P-complete.
This yields optimal complexity for the recognizability problem for tree-automatic relations:
Recognizability is \P-complete for relations given by deterministic bottom-up or top-down tree automata,
and \EXP-complete for nondeterministic tree automata.

\paragraph*{Contributions}

\begin{figure*}
\centering
\begin{tikzpicture}[yscale=0.8,anchor=base, baseline, ->,>=stealth]
	\node (Rec) at (0,0) {$\Rec$};
	\node[rotate=90] at (0,1) {$\subsetneq$};
	\node (Sync) at (0,2) {$\Sync$};
	\node[rotate=90] at (0,3) {$\subsetneq$};
	\node (DRat) at (0,4) {$\DRat$};
	\node[rotate=90] at (0,4.8) {$\subsetneq$};
	\node (Rat) at (0,5.3) {$\Rat$};
	
	\draw[bend right = 60] (Sync) to node [align = right, left] {\small\NL-c. for DFAs \\ \small\PSPACE-c. for NFAs \\ \small(see \cite{BHLLN19})} (Rec);
	\draw[bend right = 60] (DRat) to node [align = right, left] {\small$\P$ for $k = 2$ \\ \small$(2k-4)$-$\EXP$ for $k > 2$ \\ \small(see \Cref{thm:drat-sync})} (Sync);
	\draw[bend left = 40] (DRat) to node [align = right, right] {\small$\P$ for $k = 2$ \\ \small$\coREXP$ for $k > 2$ \\ \small(see \Cref{thm:drat-recognizability})} (Rec);
\end{tikzpicture}\hspace{2em}
\begin{tikzpicture}[yscale=0.8,anchor=base, baseline, ->,>=stealth]
	\node (Rec) at (0,0) {$\wRec$};
	\node[rotate=90] at (0,1) {$\subsetneq$};
	\node (Sync) at (0,2) {$\wSync$};
	\node[rotate=90] at (0,3) {$\subsetneq$};
	\node (DRat) at (0,4) {$\wDRat$};
	\node[rotate=90] at (0,4.8) {$\subsetneq$};
	\node (Rat) at (0,5.3) {$\wRat$};
	
	\tikzstyle{open} = [dotted]
	\draw[bend right = 60] (Sync) to node [align = right, left] {\small\NL-c. for DPAs \\ \small\PSPACE-c. for NBAs \\ \small(see \Cref{thm:omega-main}) } (Rec);
	\draw[bend right = 60, open] (DRat) to node [align = right, left] {\small open} (Sync);
	\draw[bend left = 40, open] (DRat) to node [align = right, right] {\small open} (Rec);
\end{tikzpicture}

\caption{The complexity landscape of deciding membership to subclasses of rational relations over finite and infinite words.
An arrow from $\C_2$ to $\C_1$ refers to the membership problem, given a relation from $\C_2$, does it belong to $\C_1$?
Membership of rational relations in any one of the three subclasses is undecidable~\cite{FischerR68,Lisovik79}.
Dotted arrows mean that decidability of the problem is unknown. }
\label{fig:results}

\end{figure*}

We provide improved complexity and new decidability results for the membership problems
in subclasses of rational relations over finite and infinite words.
To do so, we refine the existing analyses in \cite{CartonCG06,LodingS19} and identify patterns in the transducers
which witness \emph{non}-membership in the subclass.
As our first main result, we pinpoint the precise complexity of the $\omega$-recognizability problem of $\omega$-synchronous relations.

\begin{theorem}\label{thm:omega-main}
Given an $\omega$-synchronous relation $R$ by a deterministic parity (resp. nondeterministic Büchi) automaton,
it is \NL-complete (resp. \PSPACE-complete) to decide whether $R$ is $\omega$-recognizable.
\end{theorem}

This matches the complexity of the recognizability problem of synchronous relations over finite words~\cite{BHLLN19}.
To prove \cref{thm:omega-main}, we follow the approach of~\cite{LodingS19} and solve the infinite clique problem
for $\omega$-synchronous co-equivalence relations $\not \sim$.
Their algorithm constructs an automaton for a regular set of (ultimately periodic) representatives of $\sim$,
whose size is double (triple) exponential in the size of a (non)deterministic automaton for $\not \sim$.
We circumvent the construction of this large automaton and identify a simple pattern directly in the automaton for $\not \sim$
which witnesses an infinite clique.

Our second and third main result concerns decision problems on deterministic rational relations over finite words.
We encounter two issues when applying the same reduction to the infinite clique problem on $\not \sim_j^R$.
If $R$ is a binary relation, then it is not difficult to see that $\not \sim_1^R$ is effectively rational
since two runs on pairs $(x,z)$ and $(y,z)$ with a common second component $z$
can be simulated in parallel by a 3-tape automaton reading $(x,y,z)$.
However, to the best of the authors' knowledge, it is unknown
whether the infinite clique problem for rational relations is decidable at all,
even when restricted to co-equivalence relations.
Moreover, if $R$ has arity $k > 2$ then it is unclear whether the relations $\not \sim_j^R$ are still rational.
Instead of reducing to an infinite clique problem, we revisit the proof from~\cite{CartonCG06} and
obtain the following improved complexity bounds.

\begin{theorem}
	\label{thm:drat-recognizability}
	Given a $k$-ary deterministic rational relation $R$,
	one can decide whether $R$ is recognizable
	(i) in $\P$ if $k = 2$, (ii) in $\coREXP$ if $k > 2$ is fixed, and (iii) in $\coNEXP$ if $k$ is part of the input.
\end{theorem}

Here, $\coREXP \subseteq \coNEXP$ is the class of all decision problems that can be solved by a randomized algorithm in exponential time,
which may err on negative instances with probability at most $1/2$.
To show \cref{thm:drat-recognizability} we reduce the recognizability problem
to the equivalence problem of deterministic $k$-tape automata.
The reduction works in logspace if $k = 2$, but requires polynomial space for $k > 2$.
Let us remark that the precise complexity of the equivalence problem is unknown:
Harju and Karhumäki showed that testing equivalence of deterministic rational relations is in $\coNP$~\cite{HarjuK91}.
Moreover, Friedman and Greibach devised a polynomial time algorithm for binary relations~\cite{FriedmanG82},
and for fixed arity $k > 2$ equivalence is decidable in randomized polynomial time~\cite{Worrell13}.
For the reduction, we first observe that recognizability can also be described in terms
of modified equivalence relations $\approx_j^R$ over words instead of $\sim_j^R$, which is defined over $j$-tuples of words.
Second, we extract from~\cite{CartonCG06} an automaton pattern that witnesses nonrecognizability.
In addition, for the case of binary relations, we need the simple but crucial observation mentioned above
that two runs on pairs of words with a common component can be simulated in parallel.

In addition, we observe that over deterministic rational relations the equivalence problem is logspace reducible to the recognizability problem (\cref{prop:eqtorec}).
Essentially, this follows from a result by Friedman and Greibach~\cite{FriedmanG78},
which reduces the equivalence problem of DPDAs restricted to a subclass $\C$
to the membership problem of DPDAs to $\C$.
Hence, over binary deterministic rational relations the recognizability problem and the equivalence problem
are in fact \emph{logspace interreducible}.

Moreover, we present a construction that transforms
a deterministic multitape automaton into an equivalent double exponentially sized \emph{independent multitape automaton},
assuming it exists, i.e.\ if the relation is recognizable.
This provides an answer to the problem of how to compute \emph{monadic decompositions}
for deterministic rational and synchronous relations, see~\cite[Section~6]{BHLLN19}.
The construction is based on known ideas from~\cite{CartonCG06}
and imitates Valiant's construction of a double exponentially large DFA from a regular DPDA~\cite{Valiant75}.
It seems that the missing piece for the construction is
our characterization of recognizability via the equivalence relations $\approx_j^R$.

Finally, we prove that one can decide whether a deterministic rational relation is synchronous
by a reduction to the recognizability problem, which was left open in~\cite{CartonCG06}.

\begin{theorem}\label{thm:drat-sync}
Given a $k$-ary deterministic rational relation $R$,
one can decide whether $R$ is synchronous
(i) in $\P$ if $k = 2$ and
(ii) in $(2k-4)$-$\EXP$ if $k > 2$.
\end{theorem}

The intuition behind the algorithm is that the heads of a deterministic multitape automaton for a synchronous relation
must have \emph{bounded delay} throughout the computation, see \cite[Section~3]{FrougnyS93}, except if,
from some point on, the components are independent from each other.
To check the latter condition, we need the recognizability test from \cref{thm:drat-recognizability}.

\paragraph*{Applications}
As a corollary of our results on $\omega$-synchronous relations we will provide a $\PSPACE$-algorithm
which tests whether a quantifier-free formula over mixed real-integer linear arithmetic $(\R;\Z,+,<,0,1)$ is
\emph{monadically decomposable}, i.e.\ equivalent to a Boolean combination of monadic formulas.

Recently, recognizable relations have been featured in a decidable string constraint language,
motivated by the verification of string-manipulating programs~\cite{ChenHLRW19}.
One semantic condition of the constraint language requires that the relations appearing in the constraint are \emph{effectively} recognizable,
i.e.\ one can compute a representation as a union of Cartesian products of regular languages.
If the given relations are deterministic rational, we can effectively decide recognizability
(in polynomial-time for binary relations) and compute the required representation in double exponential time.

\section{Rational relations and their Subclasses}
\label{sec:classes}

In the following we introduce the classes of rational relations,
deterministic rational relations, synchronous relations, and recognizable relations,
which are denoted by $\Rec \subseteq \Sync \subseteq \DRat \subseteq \Rat$.
Similarly, on infinite words we consider $\omega$-rational relations,
deterministic $\omega$-rational relations, $\omega$-synchronous relations, and $\omega$-recognizable relations,
denoted by $\wRec \subseteq \wSync \subseteq \wDRat \subseteq \wRat$.
Since $\wDRat$ and $\wRat$ will not be used in this work, we will not define these classes.

Let $\Sigma$ be a finite alphabet.
The product of $k$ free monoids $(\Sigma^*)^k$
forms a monoid with componentwise multiplication $(u_1, \dots, u_k)(v_1, \dots, v_k) = (u_1v_1, \dots, u_kv_k)$.
We often denote word tuples by boldface letters $\bm{u}$ and denote its $i$-th entry by $u_i$.
As usual we identify a pair of tuples $(\bm{u},\bm{v})$ with the concatenation of $\bm{u}$ and $\bm{v}$.
Furthermore $\bm{\varepsilon} = (\varepsilon, \dots, \varepsilon)$ denotes a tuple of empty words
of appropriate dimension.
The \emph{length} of a word tuple $\|\bm{u}\| = \sum_{i=1}^k |u_i|$ is the total length of its entries.
We assume familiarity with the basic models of (non)deterministic finite automata over finite and infinite words.
Recall that the class of $\omega$-regular languages is described by \emph{nondeterministic Büchi automata} (NBAs)
as well as by \emph{deterministic parity automata} (DPAs)~\cite[Section~1]{GTW02}.

\paragraph{Rational relations}
A $k$-tape automaton $\A = (Q,\Sigma,q_0,\Delta,F)$ consists of a finite state set $Q$,
a finite alphabet $\Sigma$, an initial state $q_0$, a set $F \subseteq Q$ of final states,
and a finite set of transitions $\Delta \subseteq Q \times (\Sigma^*)^k \times Q$.
A run of $\A$ on a tuple $\bm{w} \in (\Sigma^*)^k$ from $p_0$ to $p_n$ is a sequence of transitions
$p_0 \xrightarrow{\bm{w}_1} p_1 \xrightarrow{\bm{w}_2} \cdots \xrightarrow{\bm{w}_n} p_n$
with $\bm{w} = \bm{w}_1 \bm{w}_2 \cdots \bm{w}_n$.
The relation $R(\A)$ accepted by $\A$ consists of all tuples $\bm{w} \in (\Sigma^*)^k$
such that $\A$ has a run on $\bm{w}$ from the initial to a final state.
Relations accepted by $k$-tape automata are called \emph{rational}.

\paragraph{Deterministic rational relations}
For $k$-tape automata we define the sets $H_1, \dots, H_k$
by $H_i = \{\varepsilon\}^{i-1} \times \Sigma \times \{\varepsilon\}^{k-i}$.
A $k$-tape automaton $\A = (Q,\Sigma,q_0,\Delta,F)$ is \emph{deterministic}
if (i) $Q$ is equipped with a partition into sets $Q = \bigcup_{i=1}^k Q_i$,
(ii) the transition relation has the form $\Delta \subseteq \bigcup_{i=1}^k Q_i \times H_i \times Q$, and
(iii) for every $(p,h) \in Q_i \times H_i$ there exists exactly one transition $(p,h,q) \in \Delta$.
For convenience, we represent $\Delta$ as a transition function $\delta \colon Q \times \Sigma \to Q$ instead.
Observe that 1-tape (deterministic) automata are precisely NFAs (DFAs).
A relation $R \subseteq (\Sigma^*)^k$ is \emph{deterministic rational} if there exists a
deterministic $k$-tape automaton $\A$ such that $R(\A) = \{ (w_1 \mathrm{\dashv}, \dots, w_k \mathrm{\dashv}) \mid (w_1, \dots, w_k) \in R\}$
where $\mathrm{\dashv} \notin \Sigma$ is a fresh endmarker.

\paragraph{Synchronous relations}
Let $\bot \notin \Sigma$ be a fresh padding symbol.
For finite words $w_1, \dots, w_k \in \Sigma^*$ with $w_i = a_{i,1} \cdots a_{i,n_i}$
we define the convolution $w_1 \otimes \dots \otimes w_k$ of length $n = \max \{n_1, \dots, n_k\}$ by
\[w_1 \otimes \dots \otimes w_k \define = \begin{bmatrix}w_1 \\ \vdots \\ w_k\end{bmatrix} 
\define = \begin{pmatrix}a_{1,1}' \\ \vdots \\ a_{k,1}'\end{pmatrix} \cdots \begin{pmatrix}a_{1,n}' \\ \vdots \\ a_{k,n}'\end{pmatrix} \in ((\Sigma \cup \{\bot\})^k)^*\]
where $a_{i,j}' = a_{i,j}$ if $j \le n_i$ and $a_{i,j}' = \bot$ otherwise.
Similarly, we define the convolution for infinite words where the padding symbol $\bot$ is not needed.
A relation $R$ over (in)finite words is \mbox{($\omega$-)}\emph{synchronous}
if $\otimes R \define = \{ w_1 \otimes \dots \otimes w_k \mid (w_1,\dots,w_k) \in R \}$ is a ($\omega$-)regular language.
A ($\omega$-)synchronous relation $R$ is always given by a finite \mbox{($\omega$-)}automaton for the language $\otimes R$.

\paragraph{Recognizable relations}
A $k$-ary relation $R$ on (in)finite words is ($\omega$-)\emph{recognizable}
if it is a finite union $R = \bigcup_{i=1}^n L_{i,1} \times \dots \times L_{i,k}$
of Cartesian products of ($\omega$-)regular languages $L_{i,j}$.

\begin{example}
\label{ex:transd}
The relation $R_1 = \{(x,y) \mid |x| + |y| \ge 2 \}$ over some finite alphabet $\Sigma$ is recognizable,
since it can be written as the union of the Cartesian products $\Sigma^{\ge 2} \times \Sigma^*$,
$\Sigma^{\ge 1} \times \Sigma^{\ge 1}$, and $\Sigma^* \times \Sigma^{\ge 2}$
where $\Sigma^{\ge \ell}$ contains all words of length at least $\ell$.
The equality relation $R_2 = \{(x,x) \mid x \in \Sigma^*\}$ is synchronous but clearly not recognizable.
The relation $R_3 = \{(x,y) \mid x \text{ is a scattered subword of } y \}$ is deterministic rational 
(the deterministic automaton greedily embeds $x$ into $y$) but not synchronous.
The relation $R_4 = \{(x,y) \mid x \text{ is an infix of } y \}$ is rational but not deterministic rational.
\end{example}

\section{Deciding Recognizability via finite-index equivalences}
\label{sec:finite-index}

The key to decide recognizability of relations is a characterization by equivalence relations,
akin to the Myhill-Nerode equivalence for languages.
Let $D$ be any domain (e.g.\ $\Sigma^*$ or $\Sigma^\omega$) and let $R \subseteq D^k$ be a $k$-ary relation.
For $I \subseteq \{1, \dots, k\}$, and two tuples $\bm{u} \in D^{|I|}$, $\bm{v} \in D^{k-|I|}$
we define $\bm{u} \odot_I \bm{v} \in D^k $ to be the unique $k$-tuple
whose projection to $I$ is $\bm{u}$ and whose projection to $\{1, \dots, k\} \setminus I$ is $\bm{v}$.
Define the equivalence relation $\approx_I^R$ on $D^{|I|}$ by
\begin{align*}
	\bm{x} \approx_I^R \bm{y} \define \iff & \text{for all } \bm{z} \in D^{k-|I|} \colon \\
	& (\bm{x} \odot_I \bm{z} \in R \iff \bm{y} \odot_I \bm{z} \in R).
\end{align*}
Usually, $R$ will be clear from the context and we simply write $\approx_I$.
If $I = \{j\}$ is a singleton we also write $\odot_j$ and $\approx_j$ instead of $\odot_I$ and $\approx_I$.
Notice that $\approx_{[1,j]}$ coincides with the relation $\sim_j$ from the introduction.
For example, if $R_1$ is the relation from \cref{ex:transd}, then $\approx_1^{R_1}$ has three equivalence classes
$\Sigma^{\ge 2}$, $\Sigma$ and $\{\varepsilon\}$.
We need the following characterization of recognizable relations.
\begin{proposition}
	\label{prop:rat-fin-idx}
	Let $R \in \Rat \, \cup \, \wSync$. The following are equivalent:
	\begin{enumerate}
	\item $R$ is \mbox{($\omega$-)}recognizable.
	\item $\approx_{[1,j]}^R$ has finite index for all $j \in [1,k-1]$.
	\item $\approx_j^R$ has finite index for all $j \in [1,k-1]$.
	\end{enumerate}
\end{proposition}
Equivalence of 1) and 2) was proved in \cite[Proposition~3.8]{CartonCG06} for rational relations
and in \cite[Lemma~3]{LodingS19} for $\omega$-synchronous relations.
In the following we prove equivalence of 2) and 3) by applying a result by Cosmadakis, Kuper, and Libkin~\cite{CosmadakisKL01}.
We say that $R \subseteq D^k$ is an \emph{$I$-relation} if there exists a relation $S \subseteq D^{|I|}$
such that $R = \{ \bm{u} \odot_I \bm{v} \mid \bm{u} \in S, ~ \bm{v} \in D^{k-|I|} \}$.
Let $P$ be a partition of $\{1, \dots, k\}$.
We say that $R$ \emph{conforms} to $P$ if $R$ is a finite Boolean combination of relations $R_1, \dots, R_n$
where each $R_i$ is an $I$-relation for some $I \in P$.
For example, a relation $R$ conforms to the \emph{discrete partition} $\{\{1\}, \dots, \{k\}\}$ if and only if
$R$ is a finite Boolean combination of Cartesian products $L_1 \times \cdots \times L_k$
of sets $L_i$.
The following lemma is easy to show:

\begin{lemma}
	\label{lem:approx-conform}
The equivalence relation $\approx_I^R$ has finite index if and only if $R$ conforms to $\{ I, [1,k] \setminus I \}$.
\end{lemma}

If $R$ conforms to $P$ then clearly $R$ conforms to any partition $P'$ that is coarser than $P$.
The \emph{coarsest refinement} $P_1 \sqcap P_2$ of two partitions
is the set of all nonempty intersections $I_1 \cap I_2$ where $I_1 \in P_1$, $I_2 \in P_2$.

\begin{theorem}[\cite{CosmadakisKL01}]
	\label{thm:partition}
	If $R \subseteq D^k$ conforms to two partitions $P_1, P_2$ then also to their coarsest refinement $P_1 \sqcap P_2$.
\end{theorem}

The partition of $[1,k]$ \emph{generated} by subsets $I_1, \dots, I_n \subseteq [1,k]$
is the coarsest refinement $P_1 \sqcap \cdots \sqcap P_n$ of the partitions $P_j = \{ I_j, [1,k] \setminus I_j \}$.
For example, the discrete partition is clearly generated by the singleton sets $\{1\},\dots,\{k-1\}$.
It is also generated by all intervals $[1,j]$ for $j \in [1,k-1]$.

\begin{lemma}\label{lem:conforms-inf-index}
If a partition $P$ is generated by $I_1, \dots, I_n \subseteq [1,k]$ 
then $R \subseteq D^k$ conforms to $P$ if and only if $\approx_{I_j}^R$ has finite index for all $j \in [1,n]$.
\end{lemma}
\begin{proof}
By \Cref{thm:partition}, $R$ conforms to $P$ if and only if $R$ conforms to $\{ I_j, [1,k] \setminus I_j\}$ for each $j \in [1,n]$.
By \cref{lem:approx-conform} this is equivalent to finite index of $\approx_{I_j}^R$ for all $j \in [1,n]$.
\end{proof}

Choosing the discrete partition on $[1,k]$ as $P$ in \cref{lem:conforms-inf-index}
we obtain the equivalence of 2) and 3) in \cref{prop:rat-fin-idx}.

\section{Deciding $\omega$-Recognizability in $\wSync$}
\label{sec:omega}

The goal of this section is to prove \Cref{thm:omega-main}.
The lower bounds are inherited from the finite-word case by padding,
since recognizability of synchronous relations is $\PSPACE$-complete (resp. $\NL$-complete)
if the relation is given by an NFA (resp. DFA) \cite{BHLLN19}.
For the upper bounds we follow the same approach as in~\cite{BHLLN19,BGLZ22} for the recognizability problem for synchronous relations.
Given an ($\omega$-)synchronous relation $R$ the complements $\not \approx_j^R$
are again ($\omega$-)synchronous.
In fact, if $R$ is given by a (non)deterministic automaton, then a nondeterministic automaton for $\not \approx_j^R$
can be computed in logspace (polynomial space):
Observe that $x \not \approx_j^R y$ if and only if
\begin{align*}
	\exists \bm{z} \in (\Sigma^\omega)^{k-1} \colon & (x \odot_{j} \bm{z} \in R ~ \wedge ~ y \odot_{j} \bm{z} \notin R) ~ \vee \\
	& \quad (x \odot_{j} \bm{z} \notin R ~ \wedge ~ y \odot_{j} \bm{z} \in R).
\end{align*}
If $R$ is given by a DPA $\B$,
we can construct an DPA for $(\Sigma^\omega)^k \setminus R$ in logarithmic space
and convert it into an NBA $\bar \B$~\cite{KingKV01}.
If $R$ is given by an NBA $\B$ then this step incurs an exponential blowup
but can still be done in polynomial space~\cite{SVW87}.
From $\B$ and $\bar \B$ we can construct NBAs for the relations $\not \approx_j$ in logspace
(intersections, unions, and projections of NBAs are logspace computable).

By \cref{prop:rat-fin-idx} it remains to check whether for some $j < k$ the relation $\not \approx_j^R$ has an \emph{infinite clique},
i.e.\ an infinite sequence of pairwise distinct words $w_1, w_2, \dots$ such that $w_{i_1} \not \approx_j^R w_{i_2}$ for all $i_1 < i_2$.
In \cite{BGLZ22} it is shown that the infinite clique problem can be solved in nondeterministic logspace for arbitrary synchronous relations over finite words.
For arbitrary $\omega$-synchronous relations, it is a longstanding open problem whether the infinite clique problem is decidable.
However, in \cite{LodingS19} it is shown to be decidable in double exponential time for $\omega$-synchronous \emph{co-equivalence relations},
i.e.\ complements of equivalence relations.
In the following we will show that for those relations the infinite clique problem can even be solved in nondeterministic logspace.
Applying this result to the relations $\not \approx_j^R$ yields an $\NL$ respectively $\PSPACE$ algorithm
for $\omega$-recognizability of $\omega$-synchronous relations
depending on whether $R$ is given by a DPA or NBA.

\begin{theorem}
	\label{thm:omega-infinite-clique}
	It is $\NL$-complete to decide,
	given a nondeterministic Büchi automaton for an $\omega$-synchronous co-equivalence relation $\bar E$,
	 whether $\bar E$ has an infinite clique.
\end{theorem}

The rest of this section is devoted to proving \cref{thm:omega-infinite-clique}.
One challenge in finding infinite cliques in $\omega$-synchronous relations is
how to even finitely represent infinite clique, i.e.\ an infinite sequence of infinite words.
A strong indicator that this is indeed difficult is that there are
$\omega$-synchronous relations which have infinite cliques but no \emph{regular} infinite clique.
One such example is the complement of the \emph{equal ends} relation
$\ee$ on $\Sigma^\omega$ where $u \ee v$
if and only if there exist $x,y \in \Sigma^*$, $z \in \Sigma^\omega$ with $|x| = |y|$ and $u = xz$ and $v = yz$.
Kuske and Lohrey observed that, although $\not \ee$ has infinite cliques it does not have regular infinite cliques~\cite[Example~2.1]{KuskeL08}.

Instead of checking whether $\bar E$ has an infinite clique,
we follow the approach of Löding and Spinrath \cite{LodingS19} and equivalently decide
whether $\bar E$ has \emph{unbounded cliques}: For each $n \ge 1$ there exists a clique $(w_1, \dots, w_n)$ of size $n$ in $R$,
i.e.\ $(w_i,w_j) \in \bar E$ for all $1 \le i < j \le n$.
The important observation made in~\cite{LodingS19} is that it suffices to search for unbounded cliques consisting of ultimately periodic words.
Since ultimately periodic words $uv^\omega$ can be encoded by the finite word $u \# v$,
this allows us to reduce the infinite clique problem over $\omega$-synchronous co-equivalence relations
to a question over synchronous relations over finite words.

For the rest of this section fix an NBA $\A = (Q,\Sigma^2,q_0,\Delta,F)$
for the complement $\bar E \subseteq \Sigma^\omega \times \Sigma^\omega$ of an equivalence relation $E$.
Let us recall the results from \cite{LodingS19} that will be used for our proof.
Define the equivalence relation $E_\# \subseteq (\Sigma^* \# \Sigma^*)^2$ where $\# \notin \Sigma$ by
\[
E_\# \define = \{(u \# v, x \# y) \mid (uv^\omega,xy^\omega) \in E, ~ |u| = |x|, ~ |v| = |y| \},
\]
which is an $\omega$-synchronous relation~\cite{LodingS19}.
Let $L_\#(E)$ be a regular set of representatives of $E_\#$,
e.g.\ consisting of the length-lexicographically minimal elements in each class~\cite[Proof~of~Proposition~3.9]{CartonCG06}.
A language $L \subseteq \Sigma^*$ is \emph{slender} if there exists a $k \in \N$ such that
for all $\ell \in \N$ it holds that $|L \cap \Sigma^\ell| < k$.
The following lemma is shown in \cite[Lemmas~12~and~13]{LodingS19}.

\begin{lemma}\label{lem:decomposition}
One can write $L_\#(E) = \bigcup_{(i,j) \in I} P_i \{\#\} S_j$ for a finite index set $I \subseteq \N^2$ 
and non-empty regular languages $P_i,S_j \subseteq \Sigma^*$ for all $(i,j) \in I$.
Furthermore, $E$ has finite index if and only if $P_i$ and $S_j$ are slender for all $(i,j) \in I$.
\end{lemma}

The approach in \cite{LodingS19} for checking whether $E$ has finite index
is to construct automata for each of the languages $P_i$ and $S_j$ and check whether they are slender.
However, an automaton for the set of representatives $L_\#(E)$ might be double exponentially large,
since its size is exponential in the size of the automaton for $E_\#$,
which in turn is exponential in the size of the automaton for $\bar E$ via a construction using transition profiles.
This results in a double exponential time algorithm given an automaton for $\bar E$.
We deviate from this approach and use the slenderness property of the languages $P_i$ and $S_j$ only to identify
the shape of unbounded cliques in $\bar E$.
In a second step, we search for patterns in the automaton for $\bar E$ that witness unbounded cliques in $\bar E$. 
The existence of these patterns can be checked in nondeterministic logspace given an automaton for $\bar E$.

\begin{restatable}{lemma}{notslender}
\label{lem:2-cycles}
A regular language $L \subseteq \Sigma^*$ is not slender if and only if
there are words $u,v,w,x,y \in \Sigma^*$ with $|v| = |w| = |x| > 0$ and $v \ne w$
such that $u v^* w x^* y \subseteq L$.
\end{restatable}
\begin{proof}
Consider the minimal trimmed DFA for $L$.
By \cite[Theorem~4.23]{Pin22} $L$ is not slender if and only if the DFA contains two distinct nonempty cycles
$p \xrightarrow{v} p$, $q \xrightarrow{x} q$ (distinct means that their set of transitions are distinct),
and a run $p \xrightarrow{w} q$.
In particular, $u v^* w x^* y$ is contained in $L$
where $u$ and $y$ are words read from the initial state to $p$, and from $q$ to some final state.
We can ensure that $|w| \le |v| = |x|$ by replacing $v$ and $x$ by $v^{k|x|}$ and $x^{k|v|}$, respectively,
for a sufficiently large number $k$.
Furthermore, we can ensure $|v| = |w| = |x|$ by extending the run $p \xrightarrow{w} q$
to $p \xrightarrow{wx_1} r$ for some prefix $x_1$ of $x$ with $|w x_1| = |x|$
and rebasing the cycle $q \xrightarrow{x} q$ on state $r$.
\end{proof}

The following lemma distinguishes two types of cliques:
On the one hand, there are cliques whose words differ in a finite prefix but have equal ends;
on the other hand, there are cliques whose words do not have equal ends.
For example, consider the equality relation $=$.
In the complement $\neq$ we can find the cliques $(a^i b^{n-i} a^\omega)_{0 \le i \le n}$ for all $n \in \N$
of words that only differ in a finite prefix.
On the other hand, if we consider the equal ends $\ee$ equivalence relation
then we observe that it does not suffice to look at the finite prefixes of words to determine whether they are in relation or not.
In the complement $\not \ee$ we can find the cliques $((a^i b^{n-i})^\omega)_{0 \le i \le n}$ for all $n \in \N$ of words
that differ in the periodic part.

\begin{lemma}\label{lem:unbounded-cliques}
$\bar E$ contains an infinite clique if and only if $\bar E$ contains cliques of the form
$(u v^i w x^{n-i} y z^\omega)_{0 \le i \le n}$ or
$(z (u v^i w x^{n-i} y)^\omega)_{0 \le i \le n}$
for all $n \in \N$ where $u,v,w,x,y,z \in \Sigma^*$ with $|v| = |w| = |x| > 0$ and $v \ne w$.
\end{lemma}
\begin{proof}
The ``if'' direction is immediate since the existence of such cliques imply that $\bar E$ contains unbounded cliques
and therefore also an infinite clique.

For the ``only if'' direction assume that $\bar E$ contains an infinite clique which means that $E$ has infinite index.
Then by \Cref{lem:decomposition} there are non-empty regular languages $P,S \subseteq \Sigma^*$
with $P \{\#\} S \subseteq L_\#(E)$ such that $P$ or $S$ is not slender.
By \Cref{lem:2-cycles} there are $u,v,w,x,y \in \Sigma^*$ with $|v| = |w| = |x| > 0$ and $v \ne w$
such that $uv^*wx^*y \subseteq P$ or $uv^*wx^*y \subseteq S$.
If $uv^*wx^*y \subseteq P$ we pick a word $z \in S$ from the non-empty language $S$.
Then $u v^* w x^* y \# z \subseteq L_\#(E)$.
Since all words of the form $u v^i w x^j y$ are pairwise different,
$\bar E$ contains the clique $(u v^i w x^{n-i} y z^\omega)_{0 \le i \le n}$ for all $n \in \N$.
Similarly, if $uv^*wx^*y \subseteq S$, we pick a word $z \in P$
and find the cliques $(z (u v^i w x^{n-i} y)^\omega)_{0 \le i \le n}$ of size $n$ in $\bar E$.
\end{proof}

A \emph{3-cycles pattern} consists of states $q_1,q_2,q_3,q_4,q_5 \in Q$
and words $u,v,w,x,y \in \Sigma^*$ with $|v| = |w| = |x| > 0$ and $v \ne w$ such that
\begin{gather*}
q_1 \xrightarrow{\bsvector{u}{u}} q_2, \quad 
q_2 \xrightarrow{\bsvector{v}{v}} q_2, \quad 
q_2 \xrightarrow{\bsvector{w}{v}} q_3, \quad
q_3 \xrightarrow{\bsvector{x}{v}} q_3, \\
q_3 \xrightarrow{\bsvector{x}{w}} q_4, \quad
q_4 \xrightarrow{\bsvector{x}{x}} q_4, \quad
q_4 \xrightarrow{\bsvector{y}{y}} q_5.
\end{gather*}
We say that the above is a \emph{3-cycles pattern from $q_1$ to $q_5$}.
The 3-cycles pattern is called \emph{final} if one of the runs
$q_1 \xrightarrow{\bsvector{u}{u}} q_2$, $q_2 \xrightarrow{\bsvector{w}{v}} q_3$,
$q_3 \xrightarrow{\bsvector{x}{w}} q_4$, $q_4 \xrightarrow{\bsvector{y}{y}} q_5$
visits a final state.
If there exists a (final) 3-cycles pattern from $p$ to $q$
we write $q_1 \threecp q_5$ and
$q_1 \threecpf q_5$, respectively.

Clearly $p \threecp q$ implies that the automaton contains $p$-$q$-runs reading $uv^iwx^{n-i}y \otimes uv^jwx^{n-j}y$
for all $i < j \le n$, for some $u,v,w,x,y$ with $|v| = |w| = |x| > 0$ and $v \ne w$.
To prove that the converse also holds, we use transition profiles.
A \emph{transition profile} $\tau = (\Rightarrow,\stackrel{F}{\Rightarrow})$ over $\A$ consists of two binary relations
$\Rightarrow, \stackrel{F}{\Rightarrow}$ over $Q$.
For each word $w \in (\Sigma^2)^*$ we define the transition profile $\tau(w)$ such that
$p \Rightarrow q$ if and only if there exists a run $p \xrightarrow{w} q$, and
$p \stackrel{F}{\Rightarrow} q$ if and only if there exists a run $p \xrightarrow{w} q$ visiting a final state.
It is easy to see that $\tau(uv)$ is determined by $\tau(u)$ and $\tau(v)$, and therefore the set $\mathsf{TP}(\A) = \{ \tau(w) \mid w \in (\Sigma^2)^* \}$
forms a finite monoid with the well-defined operation $\tau(u) \cdot \tau(v) = \tau(uv)$ and neutral element $\tau(\varepsilon)$.
An element $s$ in a monoid $M$ is \emph{idempotent} if $s^2 = s$.
Every finite monoid $M$ has an \emph{idempotent exponent}, i.e.\ a number $n \ge 1$ so that $s^n$ is idempotent
for all $s \in M$.

\begin{lemma}\label{lem:existence-3-cycles}
Let $p \in \N$ be the idempotent exponent of $\mathsf{TP}(\A)$.
If for words $u,v,w,x,y,z \in \Sigma^*$ with $|v| = |w| = |x| > 0$ and $v \ne w$ and states $q_1,q_5 \in Q$
there exists a run $\rho$ in $\A$ from $q_1$ to $q_5$ reading
\[
	\bsvector{u}{u} \bsvector{v^{pn}}{v^{pn}} \bsvector{v^{p-1}w}{v^p} \bsvector{x^{pn}}{v^{pn}} \bsvector{x^p}{v^{p-1}w} \bsvector{x^{pn}}{x^{pn}} \bsvector{y}{y}
\]
then $q_1 \threecp q_5$. If $\rho$ visits a final state then $q_1 \threecpf q_5$.
\end{lemma}

\begin{proof}
We will use the fact that $\bsvector{v^p}{v^p}$, $\bsvector{x^p}{v^p}$, and $\bsvector{x^p}{x^p}$ are idempotent in $\mathsf{TP}(\A)$.
Let us replace $v$ by $v^p$, $w$ by $v^{p-1}w$, and $x$ by $x^p$.
Now the words $\bsvector{v}{v}$, $\bsvector{x}{v}$, and $\bsvector{x}{x}$ are idempotent in $\mathsf{TP}(\A)$,
and $\rho$ becomes a run reading
\begin{equation}
	\label{eq:pn-pair}
	\bsvector{u}{u} \bsvector{v^n}{v^n} \bsvector{w}{v} \bsvector{x^n}{v^n} \bsvector{x}{w} \bsvector{x^n}{x^n} \bsvector{y}{y}.
\end{equation}
The sequence of $n+1$ states visited before and after reading each of the $n$ factors $\bsvector{v}{v}$
must contain a repeated state $q_2$, and similarly for the factors $\bsvector{x}{v}$ and $\bsvector{v}{v}$.
Therefore we find intermediate states $q_2,q_3,q_4$ so that $\rho$ has the form
\begin{equation}
\label{eq:3cp-decomposed}
\begin{array}{c c}
\rho_1 \colon q_1 \xrightarrow{\bsvector{u v^{i_1}}{u v^{i_1}}} q_2, &
\sigma_2 \colon q_2 \xrightarrow{\bsvector{v^{i_2}}{v^{i_2}}} q_2, \\
\rho_2 \colon q_2 \xrightarrow{\bsvector{v^{i_3}}{v^{i_3}} \bsvector{w}{v} \bsvector{x^{j_1}}{v^{j_1}}} q_3, &
\sigma_3 \colon q_3 \xrightarrow{\bsvector{x^{j_2}}{v^{j_2}}} q_3 \\
\rho_3 \colon q_3 \xrightarrow{\bsvector{x^{j_3}}{v^{j_3}} \bsvector{x}{w} \bsvector{x^{k_1}}{x^{k_1}}} q_4, &
\sigma_4 \colon q_4 \xrightarrow{\bsvector{x^{k_2}}{x^{k_2}}} q_4, \\
\rho_4 \colon q_4 \xrightarrow{\bsvector{x^{k_3} y}{x^{k_3} y}} q_5 &
\end{array}
\end{equation}
for some numbers $i_1,j_1,k_1,i_3,j_3,k_3 \ge 0$ and $i_2,j_2,k_2 \ge 1$.
Since $\bsvector{v}{v}$, $\bsvector{x}{v}$, and $\bsvector{x}{x}$ are idempotent in $\mathsf{TP}(\A)$,
there exist runs $\tilde \rho_1, \tilde \sigma_2, \tilde \rho_2, \tilde \sigma_3, \tilde \rho_3, \tilde \sigma_4, \tilde \rho_4$
as in \Cref{eq:3cp-decomposed} for $i_\ell = j_\ell = k_\ell = 1$ for all $\ell \in [1,3]$.
Then the five words $uv, v^3, vwx, x^3, xy$ form the required 3-cycles pattern from $q_1$ to $q_5$.

Assume that $\rho$ visits a final state.
We can ensure that the final state occurs in one of the subruns $\rho_i$ in \Cref{eq:3cp-decomposed}:
If the final state occurs in one of the cycles $\sigma_i$ then we can append the cycle $\sigma_i$ to $\rho_i$.
By the $\stackrel{F}{\Rightarrow}$-component of transition profiles we can then choose the run $\tilde \rho_i$
to visit a final state again, and therefore $q_1 \threecpf q_5$.
\end{proof}

The next lemma shows that a 3-cycles pattern can be used to detect unbounded cliques in $\bar E$ where the words differ in the finite prefix.
\begin{lemma}\label{lem:3-cycles}
$\bar E$ contains cliques $(u v^i w x^{n-i} y z^\omega)_{0 \le i \le n}$
for all $n \in \N$ where $u,v,w,x,y,z \in \Sigma^*$ with $|v| = |w| = |x| > 0$ and $v \ne w$ if and only if
there is a 3-cycles pattern in $\A$ from $q_0$ to some state $q \in Q$
such that $\bsvector{{z'}^\omega}{{z'}^\omega}$ is accepted from $q$ for some word $z' \in \Sigma^*$.
\end{lemma}
\begin{proof}
We first observe that $\bar E$ contains cliques $(u v^i w x^{n-i} y z^\omega)_{0 \le i \le n}$ as on the LHS of the lemma if and only if
\begin{equation}\label{eq:3-cycles}
\bsvector{u}{u}\bsvector{v}{v}^*\bsvector{w}{v}\bsvector{x}{v}^*\bsvector{x}{w}\bsvector{x}{x}^*\bsvector{y}{y}\bsvector{z}{z}^\omega \subseteq L(\A)
\end{equation}
for some $u,v,w,x,y,z \in \Sigma^*$ with $|v| = |w| = |x| > 0$ and $v \ne w$.
Then the ``if'' direction of the lemma follows directly.
For the ``only if'' direction assume that \Cref{eq:3-cycles} holds.
Let $n \define= |Q|$ and $p \in \N$ be the idempotent exponent of $\mathsf{TP}(\A)$.
Then there is a run of $\A$ on $\bsvector{u(v^p)^nv^{p-1}w(x^p)^{2n+1}y}{u(v^p)^{2n+1}v^{p-1}w(x^p)^ny}$
from $q_0$ to some state $q \in Q$ such that $\bsvector{z^\omega}{z^\omega}$ is accepted from $q$.
Applying \Cref{lem:existence-3-cycles} yields the desired 3-cycles pattern.
\end{proof}

The following lemma shows which pattern occurs if the words differ in the periodic part.
In \Cref{lem:3-cycles3} we will see that this pattern is also sufficient to show the existence of unbounded cliques.
\begin{lemma}\label{lem:3-cycles2}
If $\bar E$ contains cliques $(z(u v^i w x^{n-i} y)^\omega)_{0 \le i \le n}$
for all $n \in \N$ where $u,v,w,x,y,z \in \Sigma^*$ with $|v| = |w| = |x| > 0$ and $v \ne w$, then
there are states $q_1,\dots,q_\ell \in Q$ such that
\begin{itemize}
\item $q_0 \xrightarrow{\bsvector{z}{z}} q_1$,
\item $q_1 \threecp q_2 \threecp q_3 \threecp \cdots \threecp q_{\ell-1} \threecpf q_\ell$,
\item $q_k = q_\ell$ for some $k < \ell$.
\end{itemize}
\end{lemma}

\begin{proof}
Suppose that $\bar E$ contains cliques $(z (u v^i w x^{n-i} y)^\omega)_{0 \le i \le n}$
with $|v| = |w| = |x| > 0$ and $v \ne w$.
Let $t$ be the word from \Cref{eq:pn-pair}.
Since $\bsvector{z}{z}t^\omega$ is accepted by $\A$, it has an accepting run of the form
\[q_0 \xrightarrow{\bsvector{z}{z}} q_1 \xrightarrow{t} q_2 \xrightarrow{t} q_3 \xrightarrow{t} \cdots. \]
Let $m \in \N$ such that $\{q_0, \dots, q_m\} = \{q_i \mid i \in \N \}$,
i.e.\ all states $q_i$ have been visited at least once after reaching $q_m$.
Since the run visits some final state infinitely often, there exists $\ell > m$ such that
the subrun between $q_{\ell-1}$ and $q_\ell$ visits a final state.
Furthermore, there exists $k \le m$ such that $q_k = q_\ell$.
\end{proof}

\begin{lemma}\label{lem:3-cycles3}
If there are states $q_1,\dots,q_\ell \in Q$ such that
\begin{itemize}
\item $q_0 \xrightarrow{\bsvector{z}{z}} q_1$,
\item $q_1 \threecp q_2 \threecp q_3 \threecp \cdots \threecp q_{\ell-1} \threecpf q_\ell$,
\item $q_k = q_\ell$ for some $k < \ell$,
\end{itemize}
then $\bar E$ contains unbounded cliques.
\end{lemma}

\begin{proof}
Let $u_j,v_j,w_j,x_j,y_j \in \Sigma^*$ be the words of the 3-cycles pattern 
from $q_j$ to $q_{j+1}$ for $j \in [1,\ell-1]$.
Define $t_j(i,n) \define = u_j v_j^i w_j x_j^{n-i} y_j$ for all $0 \le i \le n$ and $1 \le j < \ell$.
Then
\[(t_1(i,n) \cdots t_{k-1}(i,n) \big (t_k(i,n) \cdots t_\ell(i,n) \big)^\omega)_{0 \le i \le n}\]
forms a clique in $\bar E$ for each $n \in \N$.
\end{proof}

We are now ready to prove \cref{thm:omega-infinite-clique}.
Since $\bar E$ has an infinite clique if and only if it has unbounded cliques,
by \Cref{lem:unbounded-cliques,lem:3-cycles,lem:3-cycles2,lem:3-cycles3} 
it suffices to check whether $\A$ contains the pattern in \Cref{lem:3-cycles} or the pattern in \Cref{lem:3-cycles2,lem:3-cycles3}.

First, we can check in $\NL$ whether, given states $p,q \in Q$,
there exists a word $z \in \Sigma^*$ with $p \xrightarrow{\bsvector{z}{z}} q$,
and whether there exists a word $z \in \Sigma^*$ such that $\bsvector{z^\omega}{z^\omega}$ is accepted from $q$.
Furthermore, given two states $q_1, q_5 \in Q$, we can check whether $q_1 \threecp q_5$ in $\NL$ as follows:
Construct in logspace an NFA $\A_{q_1,q_5}$ which reads a convolution $u \otimes v \otimes w \otimes x \otimes y$
with $|v| = |w| = |x|$ and $v \ne w$.
It initially guesses and stores the states $q_2,q_3,q_4$.
Then, it simulates 7 copies of $\A$ in parallel to check the existence of the runs as in the definition of a 3-cycles pattern.
Observe that $q_1 \threecp q_5$ if and only if $\A_{q_1,q_5}$ accepts some word,
hence, it suffices to check nonemptiness of $\A_{q_1,q_5}$ in \NL.
Similarly, we can test $q_1 \threecpf q_5$ in \NL.
This allows us to detect the pattern in \Cref{lem:3-cycles} and the pattern in \Cref{lem:3-cycles2,lem:3-cycles3} in $\NL$.

$\NL$-hardness follows by a reduction from the finite word case~\cite{BGLZ22} by padding.


\paragraph*{Application to monadic decomposability over $\omega$-automatic structures}
Recognizability of relations over words is closely connected to the notion of monadic decomposition.
A formula $\varphi(x_1,\dots,x_k)$ is \emph{monadically decomposable} over a logical structure $\mathfrak{A}$
if $\varphi$ is equivalent to a Boolean combination of monadic formulas, i.e.\ formulas with a single free variable~\cite{VeanesBNB17}.
While generally undecidable, Libkin provides sufficient conditions on $\mathfrak{A}$
under which the problem of testing monadic decomposability becomes decidable~\cite[Theorem~3]{Libkin03}.
Under these conditions a formula $\varphi$ defining a $k$-ary relation $R$ is monadically decomposable if and only if
$\approx_{[1,j]}^R$ has finite index for all $j < k$ \cite[Lemma~4]{Libkin03}.
In its generality the algorithm in~\cite{Libkin03} is not very efficient since it
uses an unstructured enumeration procedure to find so called \emph{definable invariant Skolem functions}.

There is a more straightforward procedure for monadic decomposability
if $\mathfrak{A}$ is $\omega$-automatic~\cite{BlumensathG00}, i.e.\ its domain and relations are given by $\omega$-synchronous automata.
In this setting we can translate $\varphi$ into an $\omega$-synchronous automaton for the relation $R$ over infinite words
defined by $\varphi$.
Then, we construct automata for the relations $\not\approx_{[1,j]}^R$
and solve the infinite clique problem by~\Cref{thm:omega-infinite-clique}.
In fact, if $\varphi$ is quantifier-free (this assumption is also made in \cite{VeanesBNB17}),
the automata for $\not\approx_{[1,j]}^R$ are constructible in polynomial space.
Combined with the $\NL$-algorithm for the infinite clique problem,
this yields $\PSPACE$-complexity for testing monadic decomposability.

An example for an $\omega$-automatic structure that satisfies the conditions of \cite[Theorem~3]{Libkin03}
is real linear arithmetic (RLA) $(\R;+,<,0,1)$.
Its extension $(\R;\Z,+,<,0,1)$ to mixed real-integer linear arithmetic (RILA)
is still $\omega$-automatic \cite{BoigelotBR97}, but it is not immediately clear whether it fulfills the conditions of \cite[Theorem~3]{Libkin03}.
However, we can use the fact that in the (standard) $\omega$-automatic presentation of RILA
ultimately periodic words are rational numbers and therefore definable in RILA.
\begin{lemma}\label{lem:up-mondec-index}
Let $\mathfrak{A}$ be an $\omega$-automatic structure with an $\omega$-automatic presentation over alphabet $\Sigma$
such that each ultimately periodic word over $\Sigma$ that represents an element in the domain is definable in $\mathfrak{A}$.
Then a formula $\varphi$ in $\mathfrak{A}$ defining the relation $R \subseteq (\Sigma^\omega)^k$ is monadically decomposable if and only if
$\approx_j^R$ has finite index for all $j \in [1,k]$.
\end{lemma}
\begin{proof}
The ``only if'' direction is clear.
For the ``if'' direction assume that $\approx_j^R$ has finite index for all $j \in [1,k]$.
It was observed in~\cite[Proof of Lemma 3]{LodingS19} that every finite-index $\omega$-synchronous equivalence relation
has a set of ultimately periodic representatives.
Let $A_1, \dots, A_k$ be such representative sets for $\approx_1^R, \dots, \approx_k^R$.
Now, $\varphi(x_1,\dots,x_k)$ is equivalent to the following Boolean combination of monadic formulas:
\[
	\bigvee_{(a_1, \dots, a_k) \in R \cap (A_1 \times \cdots \times A_k)} \bigwedge_{j=1}^k x_j \approx_j a_j
\]
The statement $x_j \approx_j a_j$ is definable since $\approx_j$ is definable
and the ultimately periodic words in $A_j$ are definable in $\mathfrak{A}$.
\end{proof}
Therefore, under the assumption in \cref{lem:up-mondec-index} (e.g.\ for RILA) we can use the same approach from above
to decide monadic decomposability for quantifier-free formulas in polynomial space.

\section{Deciding Recognizability in $\DRat$}
\label{sec:rec}

In this section we will show how to test whether a deterministic rational relation is recognizable (\Cref{thm:drat-recognizability})
and, if so, how to construct an equivalent independent multitape automaton.
Notice that we can ignore the endmarkers in the definition of $\DRat$
since a relation $R$ is recognizable if and only if
$\{ \bm{w} (\mathrm{\dashv}, \dots, \mathrm{\dashv}) \mid \bm{w} \in R\}$ is recognizable.
Hence, for the rest of this section let $R \subseteq (\Sigma^*)^k$
with $R = R(\A)$ for some deterministic $k$-tape automaton $\A = (Q,\Sigma,q_0,\delta,F)$
with $n$~states.
Furthermore we assume that all states are reachable from $q_0$.
We also write $R_q$ for the relation recognized from state $q$, i.e.\ $R_q \define = R(\A_q)$ where $\A_q = (Q,\Sigma,q,\delta,F)$.

\subsection{Witness for nonrecognizability}

To decide whether $R$ is recognizable it suffices to check
whether the equivalence relations $\approx_1, \dots, \approx_{k-1}$ have finite index by \Cref{prop:rat-fin-idx}.
To keep notation clean in the following we will focus on how to test whether $\approx_1$ has finite index.
By permuting the components of $R$ we can reduce testing finite-index of any $\approx_j$ to the $\approx_1$-case.

We provide equivalent characterizations (\cref{lem:carton}) of when $\approx_1$ has infinite index,
which will be used for the decision procedures in \cref{thm:drat-recognizability}.
The characterizations will be deduced from the proof of Lemma~3.5 in \cite{CartonCG06},
which states that, if $\approx_1$ has finite index, then any word is $\approx_1$-equivalent to a word
whose length is exponentially bounded in $n$.
We need a few definitions from \cite{CartonCG06}.
A nonempty word $v_1 \in \Sigma^+$ is {\em null-transparent} if
for all $s,t \in Q_1$ we have $s \xrightarrow{(v_1,\bm{\varepsilon})} t$ implies
$t \xrightarrow{(v_1,\bm{\varepsilon})} t$.
In other words, $v_1$ induces an \emph{idempotent} transformation on $Q_1$.
Since every element $m$ in a finite monoid has an idempotent power $m^\ell$,
every non-empty word $v_1$ has a null-transparent power $v_1^\ell$.
We call a run $s \xrightarrow{(x,\bm{z})} t$ an {\em $N$-path} if the run switches from $Q_1$ to $Q \setminus Q_1$
at most $N$ times. A nonempty word $y \in \Sigma^+$ is called {\em $N$-invisible in the context of $x \in \Sigma^*$} if 
any $N$-path $s \xrightarrow{(x,\bm{z})} t$ with $t \in Q_1$ implies $t \xrightarrow{(y,\bm{\varepsilon})} t$.

\begin{lemma}[{\cite[Lemma~3.4]{CartonCG06}}]
	\label{lem:valiant}
	Let $n$ be the number of states of $\A$ and let $u_1 \cdots u_\ell \in \Sigma^*$ be a product of $\ell$ nonempty words.
	\begin{enumerate}
		\item If $\ell > n!$ then some factor $u_{i+1} \cdots u_j$ is null-transparent.
		\item If $\ell > 2(Nn)^N$ then some factor $u_{i+1} \cdots u_j$ is $N$-invisible in the context of $u_1 \dots u_i$. \label{it:invisible}
	\end{enumerate}
\end{lemma}

We say that a set $S$ \emph{separates} two sets $X$ and $Y$ if $X \subseteq S$ and $Y \cap S = \emptyset$,
or $Y \subseteq S$ and $X \cap S = \emptyset$.
If $X$ is a singleton $\{x\}$ we also say that $S$ separates $x$ and $Y$ (similarly for $Y$).

\begin{proposition}
	\label{lem:carton}
	The following conditions are equivalent:
	
	\begin{enumerate}
	\item $\approx_1$ has infinite index. \label{it:infinite}
	
	\item There exist words $x,y,z \in \Sigma^*$ such that
	$y$ is $nn!$-invisible in the context of $x$ and $xyz \not \approx_1 xz$. \label{it:remove-invisible}
	
	\item There exist $\bm{v},\bm{w} \in (\Sigma^*)^k$ and a state $q \in Q$ such that
	$q \xrightarrow{\bm{v}} q$,
	$v_1$ is null-transparent,
	$R_q$ separates $\bm{w}$ and $(v_1, \bm{\varepsilon}) \bm{w}$. \label{it:sep-one}
	
	\item There exist $\bm{v},\bm{w} \in (\Sigma^*)^k$ and a state $q \in Q$ such that
	$q \xrightarrow{\bm{v}} q$, and
	$R_q$ separates $\bm{w}$ and $(v_1, \bm{\varepsilon})^+ \bm{w}$.  \label{it:sep-many}
	\end{enumerate}
\end{proposition}

The implication (\ref{it:remove-invisible} $\Rightarrow$ \ref{it:infinite}) already appeared in \cite[Proof of Lemma~3.5]{CartonCG06}.
In our understanding, to prove this implication the authors used \ref{it:sep-one}) as an intermediate step.
Unfortunately, the proof for the implication (\ref{it:sep-one} $\Rightarrow$ \ref{it:infinite}) contains an argument
that we could not follow, see~\cref{sec:app-ccg} for a discussion.
For completeness, we reprove the implication (\ref{it:sep-one} $\Rightarrow$ \ref{it:infinite})
using \ref{it:sep-many}) as an intermediate step.

\begin{proof}[Proof of \cref{lem:carton}]

Let us start with the easy directions.

\medskip

(\ref{it:sep-many} $\Rightarrow$ \ref{it:infinite}):
Consider any run $q_0 \xrightarrow{\bm{u}} q$.
Then $u_1 v_1^i w_1 \not \approx_1 u_1 v_1^{i+j} w_1$ for all $i \ge 0$, $j \ge 1$
because $R$ separates
$\bm{u} \bm{v}^i \bm{w}$ and $\bm{u} \bm{v}^i (v_1,\bm{\varepsilon})^j \bm{w}$.
Hence $\approx_1$ has infinite index.

\medskip

(\ref{it:infinite} $\Rightarrow$ \ref{it:remove-invisible}):
Assume that \ref{it:remove-invisible}) is false.
By \Cref{lem:valiant}, any word of length at least $f(nn!)$ where $f(N) \define = 2(Nn)^n$ can be written as $uvw$
where $v$ is nonempty and $nn!$-invisible in the context of $u$, and therefore $uvw \approx_1 uw$.
By repeating this argument, we obtain for any word an $\approx_1$-equivalent word of length at most $f(nn!)$.
Therefore, $\approx_1$ has finite index.

\medskip

(\ref{it:remove-invisible} $\Rightarrow$ \ref{it:sep-one}):
Assume that $xyz \not \approx_1 xz$ where $y$ is $nn!$-invisible in the context of $x$.
Choose a length-minimal tuple $\bm{t} \in (\Sigma^*)^{k-1}$ such that
\begin{equation}
	\label{eq:w-sep}
	(xyz,\bm{t}) \in R \iff (xz,\bm{t}) \notin R.
\end{equation}
Let $\rho$ be a prefix of the run on $(xyz,\bm{t})$ which reads $x$ on the first tape.
Observe that $\rho$ is not a $nn!$-path since $y$ is $nn!$-invisible in the context of $x$
and otherwise one could remove the $(y,\bm{\varepsilon})$-loop from the run, which would contradict \Cref{eq:w-sep}.
In particular, $\rho$ reads at least $nn!$ symbols from $\bm{t}$.
Consider the sequence of states in $\rho$ visited after reading a symbol from $\bm{t}$.
There is a state $q$ which is visited more than $n!$ times.
We can factor $x = \alpha_1 \cdots \alpha_{\ell+1}$ and a prefix of $\bm{t}$ into nonempty words $\tau_1 \cdots \tau_{\ell+1}$ such that
\[
	q_0 \xrightarrow{(\alpha_1,\tau_1)} q \xrightarrow{(\alpha_2,\tau_2)} q \xrightarrow{(\alpha_3,\tau_3)} \dots \xrightarrow{(\alpha_\ell,\tau_\ell)} q
	\xrightarrow{(\alpha_{\ell+1},\tau_{\ell+1})} p
\]
and $\ell > n!$. By Lemma~\ref{lem:valiant} there exists a null-transparent factor $\alpha_{i+1} \cdots \alpha_j$ for some $1 \le i < j \le \ell$.
Let us set $x_1 = \alpha_1 \cdots \alpha_i$, $x_2 = \alpha_{i+1} \cdots \alpha_j$, and $x_3 = \alpha_{j+1} \cdots \alpha_{\ell+1}$.
Consider the corresponding decomposition $\bm{t} = \bm{t}_1 \bm{t}_2 \bm{t}_3$ such that
\begin{equation}
	\label{eq:run1}
	q_0 \xrightarrow{(x_1,\bm{t}_1)} q \xrightarrow{(x_2,\bm{t}_2)} q \xrightarrow{(x_3yz,\bm{t}_3)} r_+
\end{equation}
and
\begin{equation}
	\label{eq:run2}
	q_0 \xrightarrow{(x_1,\bm{t}_1)} q \xrightarrow{(x_2,\bm{t}_2)} q \xrightarrow{(x_3z,\bm{t}_3)} r_-
\end{equation}
where exactly one of the states $r_+, r_-$ belongs to $F$.

Since $\bm{t}$ is a length-minimal tuple satisfying \Cref{eq:w-sep} and $\bm{t}_2$ is nonempty we know that
\[
	(xyz,\bm{t}_1\bm{t}_3) \in R \iff (xz,\bm{t}_1\bm{t}_3) \in R
\]
and thus
\begin{equation}
	\label{eq:12}
	(x_2x_3yz,\bm{t}_3) \in R_q \iff (x_2x_3z,\bm{t}_3) \in R_q.
\end{equation}
We claim that either (i) $R_q$ separates $(x_2x_3yz,\bm{t}_3)$ and $(x_3yz,\bm{t}_3)$
or (ii) $R_q$ separates $(x_2x_3z,\bm{t}_3)$ and $(x_3z,\bm{t}_3)$, which proves the lemma.
Otherwise, \Cref{eq:12} implies 
\[
	(x_3yz,\bm{t}_3) \in R_q \iff (x_3z,\bm{t}_3) \in R_q,
\]
which contradicts \Cref{eq:run1,eq:run2}.
Hence, we can set $\bm{v} = (x_2,\bm{t}_2)$ and either set $\bm{w} = (x_3yz,\bm{t}_3)$ in case (i)
or set $\bm{w} = (x_3z,\bm{t}_3)$ in case (ii).
This concludes the proof.

\medskip

(\ref{it:sep-one} $\Rightarrow$ \ref{it:sep-many}):
Let $\bm{v},\bm{w}$ and $q \in Q$ such that $q \xrightarrow{\bm{v}} q$,
$v_1$ is null-transparent, $R_q$ separates $\bm{w}$ and $(v_1, \bm{\varepsilon}) \bm{w}$.
Let $m > |w_2 \cdots w_k| + 1$.
Let $\rho$ be the run on $(v_1,\bm{\varepsilon})^m \bm{w}$ starting in $q$.
It contains a subrun reading $(v_1, \bm{\varepsilon})$ between two $Q_1$-states, i.e.\ we can factor $(w_2, \dots, w_k) = \bm{x} \bm{y}$ such that
\[
	\rho \colon q \xrightarrow{(v_1^{i-1},\bm{x})} s \xrightarrow{(v_1,\bm{\varepsilon})} r \xrightarrow{(v_1^{m-i} w_1,\bm{y})} t
\]
for some $s,r \in Q_1$.
Since $v_1$ is null-transparent there is a cycle $r \xrightarrow{(v_1, \bm{\varepsilon})} r$.
Therefore $V \define= \{ (v_1,\bm{\varepsilon})^j \bm{w} \mid j \ge m \}$ is either contained in $R_q$ or disjoint from $R_q$.
Since $R_q$ separates $\bm{w}$ and $(v_1, \bm{\varepsilon})\bm{w}$, it also separates
one of them from $V$.
If $R_q$ separates $\bm{w}$ and $V$, then $\bm{v}^m$ and $\bm{w}$ satisfy the condition from the proposition.
Otherwise, $R_q$ separates $(v_1,\bm{\varepsilon}) \bm{w}$ and $V$,
and the tuples $\bm{v}^m$ and $(v_1,\bm{\varepsilon}) \bm{w}$ satisfy the condition from the proposition.
\end{proof}


\subsection{Polynomial-time algorithm for binary relations}

\begin{figure}
\centering
\begin{tikzpicture}[xscale=0.8,yscale=0.8,anchor=base, baseline, ->,>=stealth]

\tikzstyle{state}=[draw, circle, inner sep = 0, minimum size = 12pt]

\node[state] at (0,0) (q) {\footnotesize $q$};
\node[state] at (3,0) (r) {\footnotesize $r$};
\node[state, accepting] at (0,-2) (f) {\footnotesize $s$};
\node[state] at (3,-2) (g) {\footnotesize $t$};

\draw (q) to node[above] {\footnotesize $(v_1,x)$} (r);
\draw[loop above] (q) to node[above] {\footnotesize $(v_1,v_2)$} (q);
\draw[loop above] (r) to node[above] {\footnotesize $(v_1,\varepsilon)$} (r);

\draw (q) to node[right] {\footnotesize $(w_1,xy)$} (f);
\draw (r) to node[right] {\footnotesize $(w_1,y)$} (g);

\end{tikzpicture}

\caption{The pattern witnessing nonrecognizability for deterministic 2-tape automata.
Here, either state $s$ is final and $t$ is nonfinal, or vice versa.}
\label{fig:drat-not-rec}

\end{figure}

From \cref{lem:carton} we can derive a pattern
which is present in $\A$ if and only if $R$ is \emph{not} recognizable.
For binary relations the pattern is visualized in \cref{fig:drat-not-rec}.
This pattern can be detected in polynomial-time by reducing to the inequivalence problem for binary deterministic rational relations.

\begin{proposition}
	\label{prop:pattern}
	$\approx_1$ has infinite index if and only if
	there exist words $v_1,w_1 \in \Sigma^*$, tuples $\bm{v}_2, \bm{x}, \bm{y} \in (\Sigma^*)^{k-1}$, and states $q,r \in Q$ such that
	\begin{enumerate}
	\item $q \xrightarrow{(v_1,\bm{v}_2)} q$, $q \xrightarrow{(v_1,\bm{x})} r$, $r \xrightarrow{(v_1,\bm{\varepsilon})} r$,
	\item $(w_1,\bm{x}\bm{y}) \in R_q \iff (w_1,\bm{y}) \notin R_r$.
	\end{enumerate}
\end{proposition}

\begin{proof}
	For the ``if'' direction observe that $R_q$ separates $(w_1,\bm{x}\bm{y})$ and $(v_1,\bm{\varepsilon})^+(w_1,\bm{x}\bm{y})$.
	Therefore $\approx_1$ has infinite index by \cref{lem:carton} point \ref{it:sep-many}).
	For the ``only if'' direction assume that $\approx_1$ has infinite index.
	Again, by \cref{lem:carton} point \ref{it:sep-many}) there exist $(v_1,\bm{v}_2),(w_1,\bm{w}_2) \in (\Sigma^*)^k$ and a state $q \in Q$ such that
	$q \xrightarrow{(v_1,\bm{v}_2)} q$, and $R_q$ separates $(w_1,\bm{w}_2)$ and $(v_1, \bm{\varepsilon})^+ (w_1,\bm{w}_2)$.
	Let $m > \|\bm{w}_2\| + 1$ and let $\ell$ be such that $v_1^\ell$ is null-transparent.
	Consider the unique run $\rho_q$ on $(v_1, \bm{\varepsilon})^{m \ell} (w_1,\bm{w}_2)$ starting from $q$.
	It must contain a subrun of the form $s \xrightarrow{(v_1, \bm{\varepsilon})^\ell} r$ where $s,r \in Q_1$.
	Hence we can factorize $\bm{w}_2 = \bm{x}\bm{y}$ such that $\rho_q$ has the form
	\begin{equation}
		\label{eq:rhoq}
		\rho_q \colon ~ q \xrightarrow{(v_1^{(i-1) \ell},\bm{x})} s \xrightarrow{(v_1, \bm{\varepsilon})^\ell} r \xrightarrow{(v_1^{(m-i) \ell} w_1, \bm{y})} t.
	\end{equation}
	Since $v_1^\ell$ is null-transparent there exists a cycle $r \xrightarrow{(v_1, \bm{\varepsilon})^\ell} r$.
	Since $\A$ is deterministic, this allows us to choose $i = m$ in \Cref{eq:rhoq} and write
	\begin{equation}
		\rho_q \colon ~ q \xrightarrow{(v_1^{(m-1) \ell},\bm{x})} s \xrightarrow{(v_1, \bm{\varepsilon})^\ell} r \xrightarrow{(w_1, \bm{y})} t.
	\end{equation}
	Since $R_q$ separates $(w_1,\bm{w}_2) = (w_1,\bm{x}\bm{y})$ and $(v_1, \bm{\varepsilon})^{m \ell} (w_1,\bm{w}_2)$,
	we know that $(w_1,\bm{x}\bm{y}) \in R_q$ if and only if $(w_1,\bm{y}) \notin R_r$.
	Hence the words $v_1^{m \ell}, w_1$ together with the tuples $\bm{v}_2^{m \ell}, \bm{x}, \bm{y}$ satisfy the claim.
\end{proof}

\begin{theorem}
	\label{thm:binary-rec-to-eq}
	The recognizability problem for binary deterministic rational relations is logspace reducible to the equivalence problem for binary deterministic rational relations.
\end{theorem}

\newcommand{\seq}{\mathsf{seq}}

\begin{proof}
	Let $R$ be a binary deterministic rational relation,
	which is recognizable if and only if $\approx_1$ has finite index by \Cref{prop:rat-fin-idx}.
	By \cref{prop:pattern} this holds if and only if
	for all state pairs $q,r \in Q$ and all words $v_1,w_1,v_2,x,y \in \Sigma^*$ the following two conditions are equivalent:
	\begin{enumerate}
	\item[(C1)] $q \xrightarrow{(v_1,v_2)} q$, $q \xrightarrow{(v_1,x)} r$, $r \xrightarrow{(v_1,\varepsilon)} r$, $(w_1,xy) \in R_q$
	\item[(C2)] $q \xrightarrow{(v_1,v_2)} q$, $q \xrightarrow{(v_1,x)} r$, $r \xrightarrow{(v_1,\varepsilon)} r$, $(w_1,\bm{y}) \in R_r$
	\end{enumerate}
	Using an appropriate encoding we can reduce the equivalence of (C1) and (C2)
	to the equivalence problem for binary deterministic rational relations.

	Suppose $\pi,\rho$ are runs which read the same input word (in our case, this would be $v_1$), i.e.\ we can write
	\[
		\pi \colon ~ s_1 \xrightarrow{g_0} t_1 \xrightarrow{a_1} s_2 \xrightarrow{g_1} t_2 \xrightarrow{a_2} \cdots \xrightarrow{a_n} s_n \xrightarrow{g_n} t_n
	\]
	and
	\[
		\rho \colon ~ s'_1 \xrightarrow{h_0} t'_1 \xrightarrow{a_1} s'_2 \xrightarrow{h_1} t'_2 \xrightarrow{a_2} \cdots \xrightarrow{a_n} s'_n \xrightarrow{h_n} t'_n
	\]
	where each $a_i$ is a letter and the states $t_i$, $t_i'$ are precisely the states in $\pi$ and $\rho$ in $Q_1$.
	We define their synchronized shuffle $\pi \shuffle \rho \in (\Sigma \cup \{\diamond\})^*$ as
	\[
		\pi \shuffle \rho = g_0 \diamond h_0 \diamond a_1 \diamond g_1 \diamond h_1 \diamond a_2 \diamond \cdots \diamond a_n \diamond g_n \diamond h_n.
	\]
	We encode (C1) as the binary relation
	\begin{align*}
		C_1 = \{ (q \, r \, w_1, & (\pi \shuffle \rho) \, \$ \, y) \mid q,r \in Q, \, \pi \colon q \xrightarrow{(v_1,v_2)} q, \\
		& \rho \colon q \xrightarrow{(v_1,x)} r, \, r \xrightarrow{(v_1,\varepsilon)} r, \, (w_1,xy) \in R_q \}
	\end{align*}
	and (C2) as the binary relation
	\begin{align*}
		C_2 = \{ (q \, r \, w_1, & (\pi \shuffle \rho) \, \$ \, y) \mid q,r \in Q, \, \pi \colon q \xrightarrow{(v_1,v_2)} q, \\
		& \rho \colon q \xrightarrow{(v_1,x)} r, \, r \xrightarrow{(v_1,\varepsilon)} r, \, (w_1,y) \in R_r \}.
	\end{align*}
	Observe that $C_1 = C_2$ if and only if (C1) and (C2) are equivalent.
	It remains to verify that $C_1$ and $C_2$ are deterministic rational and we can construct automata in logspace.
	First, for each state pair $q,r \in Q$ we can construct a DFA over $\Sigma \cup \{\diamond\}$
	which accepts precisely the synchronized shuffles $\pi \shuffle \rho$ where
	$\pi \colon q \xrightarrow{(v_1,v_2)} q$, $\rho \colon q \xrightarrow{(v_1,x)} r$
	and $r \xrightarrow{(v_1,\varepsilon)} r$ for some words $v_1,v_2,x$.
	Since $x$ can be easily extracted as a subword of $\pi \shuffle \rho$,
	a deterministic transducer can verify whether the input pair $(q \, r \, w_1, (\pi \shuffle \rho) \, \$ \, y)$
	satisfies $(w_1,xy) \in R_q$ and whether it satisfies $(w_1,y) \in R_r$.
\end{proof}


\subsection{Arbitrary relations}

The approach from \cref{thm:binary-rec-to-eq} does not work for arity $k \ge 3$.
The issue is that the words $v_2,x,y$ from (C1) and (C2) would become $(k-1)$-tuples $\bm{v}_2, \bm{x}, \bm{y}$.
It is not clear how to appropriately encode the runs on $(v_1,\bm{x})$ and $(w_1,\bm{x}\bm{y})$ in (C1)
so that they can be simulated by an automaton.
Still, we can express (the negation of) property \ref{it:sep-one}) in \cref{lem:carton}
as the equivalence of two polynomial space constructible deterministic multitape automata.
Since equivalence of deterministic $k$-tape automata is in $\coNP$~\cite{HarjuK91}, and in $\coRP$ for fixed $k$~\cite{Worrell13},
the complexity bounds from \cref{thm:drat-sync} follow.

\begin{theorem}
	The recognizability problem for $k$-ary deterministic rational relations is polynomial space reducible to the equivalence problem for $k$-ary deterministic rational relations.
\end{theorem}

\begin{proof}
	Let $R$ be a $k$-ary deterministic rational relation,
	which is recognizable if and only if $\approx_j$ has finite index for all $j < k$ by \Cref{prop:rat-fin-idx}.
	It suffices to show how to reduce the test whether $\approx_1$ has finite index
	to the equivalence problem of polynomial space constructible deterministic $k$-tape automata.
	By \cref{lem:carton} point \ref{it:sep-one}), $\approx_1$ has finite index if and only if
	for all states $q \in Q$ and all tuples $\bm{v},\bm{w} \in (\Sigma^*)^k$ the following two conditions are equivalent:
	\begin{enumerate}
	\item[(P1)] $q \xrightarrow{\bm{v}} q$, $v_1$ null-transparent, $\bm{w} \in R_q$
	\item[(P2)] $q \xrightarrow{\bm{v}} q$, $v_1$ null-transparent, $(v_1,\bm{\varepsilon})\bm{w} \in R_q$
	\end{enumerate}
	We encode (P1) and (P2) as deterministic rational relations.
	First observe that we can construct an exponentially large DFA
	for the language of all null-transparent words $v_1 \in \Sigma^+$.
	It simulates runs on $v_1$ in parallel from every state $s \in Q_1$, and verifies that
	$s \xrightarrow{(v_1,\bm{\varepsilon})} t$ implies $t \xrightarrow{(v_1,\bm{\varepsilon})} t$.
	We encode a run $\pi$ as an alternating sequence $\mathsf{flat}(\pi) \in (Q\Sigma)^*Q$ of states and input letters.
	Under this encoding, valid runs can be recognized by a polynomially sized DFA.
	Define the following $k$-ary relations
	\begin{align*}
		P_1 = \{ (q \, \mathsf{flat}(\pi) \, \$,\bm{\varepsilon}) \, \bm{w} \mid q \in Q, \, v_1 \text{ null-transparent}, \\
		\pi \colon q \xrightarrow{\bm{v}} q, \, \bm{w} \in R_q \}
	\end{align*}
	and
	\begin{align*}
		P_2 = \{ (q \, \mathsf{flat}(\pi) \, \$,\bm{\varepsilon}) \, \bm{w} \mid q \in Q, \, v_1 \text{ null-transparent}, \\
		\pi \colon q \xrightarrow{\bm{v}} q, \, (v_1,\bm{\varepsilon})\bm{w} \in R_q \}.
	\end{align*}
	Since $v_1$ can be easily extracted from $\mathsf{flat}(\pi)$, we can construct exponentially sized deterministic $k$-tape automata
	for $P_1$ and $P_2$.
	Furthermore, the conditions (P1) and (P2) are equivalent if and only if $P_1 = P_2$.
\end{proof}

\subsection{Reducing equivalence to recognizability}

Let us complement the presented algorithms for recognizability with the following ``converse direction''.
We solved the recognizability problem by reducing to the equivalence problem over deterministic rational relations
(in logspace, for binary relations).
In fact, the equivalence problem is logspace reducible to the recognizability problem (for arbitrary arity).

\begin{restatable}{theorem}{eqtorec}
\label{prop:eqtorec}
Let $k \ge 2$.
The equivalence problem for $k$-ary deterministic rational relations is logspace reducible
to the recognizability problem for $k$-ary deterministic rational relations.
\end{restatable}


\begin{proof}

Given two deterministic $k$-tape automata $\A$ and $\B$.
First we ensure that both $R(\A)$ and $R(\B)$ are finite relations, and, in particular, recognizable.
By~\cite{HarjuK91} the automata $\A$ and $\B$ are equivalent
if and only if they accept the same tuples of length at most $n-1$,
where $n$ is the total number of states in $\A$ and $\B$.
We can compute in logspace a deterministic $k$-tape automaton $\A'$ such that $R(\A') = R(\A) \cap \{ \bm{u} \in (\Sigma^*)^k \mid \|\bm{u}\| < n \}$, and analogously $\B'$ for $\B$.
The automaton $\A'$ tracks the length of the prefix tuple read so far, up to threshold $n$,
and rejects all tuples of length at least $n$.

We claim that $R(\A') = R(\B')$ if and only if
\begin{align*}
	T = \; & \{ (a^i \#, a^i \#, \bm{\varepsilon}) \mid i \in \N \} R(\A') ~ \cup \\
	 & \{ (a^i \#, a^j \#,\bm{\varepsilon}) \mid i \neq j \} R(\B')
\end{align*}
is recognizable where $a$ and $\#$ are fresh distinct letters.
Observe that a deterministic $k$-tape automaton for $T$ is logspace computable from $\A'$ and $\B'$.
If $R(\A') = R(\B')$ then
\[
	T = \{ (a^i \#, a^j \#, \bm{\varepsilon}) \mid i, j \in \N \} R(\A')
\]
is the concatenation of two recognizable relations and hence itself recognizable.
Suppose that $R(\A') \neq R(\B')$ and assume that there exists a tuple $\bm{v} \in R(\A') \setminus R(\B')$
(the case where $R(\B') \setminus R(\A') \neq \emptyset$ is similar).
If $T$ would be recognizable then $T\bm{v}^{-1} = \{ \bm{u} \mid \bm{u}\bm{v} \in T\}$ would also be recognizable.
However
\[
	T\bm{v}^{-1} = \{ (a^i \#, a^i \#, \bm{\varepsilon}) \mid i \in \N \}
\]
is clearly not recognizable.
\end{proof}

\subsection{Constructing an independent automaton}

\Cref{thm:drat-recognizability} raises the question how to translate a deterministic multitape automaton
into an equivalent automaton with independent tapes, if one exists.
Such a construction will be needed in the next section, to decide whether a deterministic multitape automaton
recognizes a synchronous relation.
Formally, an \emph{independent $k$-tape automaton} $\I$ is a tuple $\I = (\A_1,\dots,\A_k,F)$ 
consisting of DFAs $\A_i$ without final states
and  a set of state tuples $F \subseteq Q_1 \times \dots \times Q_k$,
where $Q_i$ is the state set of $\A_i$.
The relation $R(\I)$ recognized by $\I$ is the set of all tuples $(w_1,\dots,w_k)$
such that for each $i \in [1,k]$ the unique run of $\A_i$ on $w_i$ ends in a state $q_i \in Q_i$ with $(q_1,\dots,q_k) \in F$.
Note that independent multitape automata recognize exactly the relations in $\Rec$.

\begin{theorem}\label{thm:twoexp}
	Given a deterministic $k$-tape automaton for a recognizable relation $R$,
	one can compute an independent $k$-tape automaton for $R$ in double exponential time.
\end{theorem}

\begin{proof}
For each $j \in [1,k]$ define the relation $\equiv_j$ by
\[
	x \equiv_j y \define \iff \text{for all } z \in \Sigma^* \colon xz \approx_j yz,
\]
which is a right-congruence, i.e.\ $x \equiv_j y$ implies $xa \equiv_j ya$.
Let $\A = (Q,\Sigma,q_0,\delta,F)$ be a deterministic $k$-tape automaton for a recognizable relation $R$.
Suppose that $x,y,z$ are words where $y$ is $nn!$-invisible in the context of $x$.
Then $xyz \equiv_j xz$ since otherwise $xy(zz') \not \approx_j x(zz')$ for some word $z'$,
which contradicts \cref{lem:carton}.
Hence, for each word $w$ of length $f(nn!)+1$ there exists an $\equiv_j$-equivalent word $w'$
of length at most $f(nn!)$, by cutting out $nn!$-invisible factor according to \cref{lem:valiant} where $f(N) \define = 2(Nn)^n$.
Furthermore, the function $w \mapsto w'$ can be computed in double exponential time,
as remarked in~\cite[Section~8]{Valiant75}.

Hence, the independent $k$-tape automaton $(\A_1, \dots, \A_k, F)$ works as follows.
The states of $\A_j$ are words of length at most $f(nn!)$.
The initial state is the empty word $\varepsilon$.
If the current state (word) is $w$ and the next input symbol is $a \in \Sigma$,
then the next state is the word obtained from $wa$ by removing an $nn!$-invisible factor, if possible.
In this way, at each time step the reached state is a word that is $\equiv_j$-equivalent to the read prefix.
Finally, a tuple $\bm{v}$ is marked final if and only if $\bm{v} \in R$.
On input tuple $(w_1, \dots, w_k)$ each DFA $\A_j$ reaches a state $v_j$ with $v_j \equiv_j w_j$,
and therefore $(v_1, \dots, v_k) \in R$ if and only if $(w_1, \dots, w_k) \in R$.
\end{proof}

We remark that the double exponential bound in \cref{thm:twoexp} is optimal,
which can be derived from the proof by Meyer and Fischer for the double exponential succinctness gap between DPDAs and DFAs~\cite{MeyerF71}.
To keep the paper self-contained, we provide an alternative proof.

\begin{proposition}\label{prop:lowerbound}
	There exists a recognizable relation $R_n \subseteq \{0,1\}^* \times \{0,1\}^*$
	which is accepted by a deterministic 2-tape automaton with $O(n^2 \log n)$ states
	so that any independent 2-tape automaton for $R_n$ has in total at least $2^{2^{n-1}}$ states.
\end{proposition}

\begin{proof}
Let $R_n \subseteq [1,n]^* \times [1,n]^*$ be the relation containing all pairs $(u,v)$
where $|v| \le 2n$ and $v$ is a scattered subword of $u$.
Observe that $R_n$ is accepted by a deterministic 2-tape automaton with $O(n)$ states.
Then $\approx_1^{R_n}$ is \emph{Simon's congruence} with parameter $2n$~\cite{Simon75}.
Its index is finite and bounded from below by $2^{2^{n-1}}$ by~\cite[Theorem~1.2]{KarandikarKS15}.
If an independent 2-tape automaton $\I = (\A_1,\A_2,F)$ recognizes $R_n$ then
the index of $\approx_1^{R_n}$ is a lower bound for the number of states of $\A_1$.
Finally, we can replace the alphabet $[1,n]$ by codes from $\{0,1\}^{\log n}$,
increasing the automaton size by a $\log n$-factor.
\end{proof}

\section{Deciding Synchronicity in $\DRat$}
\label{sec:rec-sync}
\newcommand{\pre}{\mathrm{pre}}
\newcommand{\suf}{\mathrm{suf}}

In this section we prove \Cref{thm:drat-sync} by showing that synchronicity can be reduced to recognizability for relations in $\DRat$,
which can be solved according to \Cref{thm:drat-recognizability}.
Let us remark that there also exists a reduction in the reverse direction,
whose proof can be found in \Cref{sec:app-rec-sync}.
\begin{restatable}{proposition}{recsync} \label{thm:rec-sync}
Given a $k$-tape automaton $\A$ for a relation $R$, one can compute in logspace a $k$-tape automaton $\B$ for a relation $S$
such that $R$ is recognizable if and only if $S$ is synchronous. If $\A$ is deterministic, then so is $\B$.
\end{restatable}

\begin{figure}
\centering
\begin{tikzpicture}[yscale=-1, every node/.style={inner sep=0,outer sep=0}]
\node[draw, fill = black!20, minimum height = 10pt, minimum width = 40pt] at (0,0) {\scriptsize $u_1$};
\node[draw, fill = black!20, minimum height = 10pt, minimum width = 105pt] at (32.5pt,10pt) {\scriptsize $u_2$};
\node[draw, minimum height = 10pt, minimum width = 30pt] at (35pt,0) {\scriptsize $s_i$};
\node[draw, minimum height = 10pt, minimum width = 30pt] at (100pt,10pt) {\scriptsize $t_{i,j}$};
\node[draw, minimum height = 10pt, minimum width = 65pt] at (82.5pt,0) {\scriptsize $\bot \:\:\: \bot \:\:\: \bot \;\;\; \bot \;\;\; \bot$};
\draw[red,opacity=0.7,semithick] (85pt,-10pt) to (85pt,20pt);
\end{tikzpicture}
\caption{An asynchronous cycle can produce words $(u_1,u_2)$ with unbounded length difference.
The words $s_i$ are pairwise inequivalent words with respect to $\approx_1^R$,
separated by the words $t_{i,j}$.}
\label{fig:sync-implies-rec}
\end{figure}

In the rest of this section we show \Cref{thm:drat-sync}.
We first give an intuition for the case $k = 2$.
Suppose $R$ is given by a deterministic 2-tape automaton $\A$ with the property that every reachable cycle
$p \xrightarrow{(v_1,v_2)} p$ satisfies $|v_1| = |v_2|$.
This ensures that $\A$ has \emph{bounded delay}~\cite[Section~3]{FrougnyS93},
i.e.\ the head positions cannot be arbitrarily far apart during the computation of $\A$.
In fact, the delay is bounded by the number of states $|Q|$.
It is well-known that such an automaton recognizes a synchronous relation~\cite[Corollary~3.4]{FrougnyS93},
since letters that are ``read ahead'' on a tape can be stored in a queue of length $|Q|$.
Let us now consider the case where $\A$ contains an \emph{asynchronous} cycle of the form
$p \xrightarrow{(v_1,v_2)} p$ with $|v_1| < |v_2|$ (the case $|v_1| > |v_2|$ is symmetric).
We partition the automaton into an asynchronous part, containing all states which are reachable from an asynchronous cycle,
and a synchronous part.
While the simulation using a queue works in the synchronous part of the automaton,
the delay can become unbounded in the asynchronous part, by traversing asynchronous cycles repeatedly.

We claim that $R$ is synchronous if and only if for each state~$q$ in the asynchronous part,
the relation $R_q$ is recognizable.
If each such relation $R_q$ is recognizable and in particular synchronous, then the computation from state $q$
can be continued synchronously using a synchronous automaton for $R_q$.
For the other direction, assume that $R_q$ is not recognizable for some asynchronous state $q$.
Therefore $\approx_1^{R_q}$ has infinite index by \cref{prop:rat-fin-idx},
i.e.\ there exist words $(s_i)_{i \ge 1}$ and $(t_{i,j})_{i<j}$ such that
$R_q$ separates $(s_i,t_{i,j})$ and $(s_j,t_{i,j})$ for all $i < j$.
We claim that the Myhill-Nerode equivalence relation $\sim_{\otimes R}$ of the language $\otimes R$ of convolutions has at least $h$ classes for each $h \in \N$:
Take an asynchronous cycle $p \xrightarrow{(v_1,v_2)} p$ from which $q$ is reachable.
We can produce runs $q_0 \xrightarrow{(u_1,u_2)} q$ where the delay $|u_2|-|u_1|$ is arbitrary large.
Pick such a run where $|u_2|-|u_1| \ge \max\{|s_1|, \dots, |s_h|\}$.
Then any two words $(u_1 s_i) \otimes u_2$ and $(u_1 s_j) \otimes u_2$ for $1 \le i < j \le h$
are inequivalent with respect to $\sim_{\otimes R}$
because $\otimes R$ separates $(u_1 s_i) \otimes (u_2 t_{i,j})$ and $(u_1 s_j) \otimes (u_2 t_{i,j})$,
see the illustration in~\cref{fig:sync-implies-rec}.
Thus, the synchronicity problem can be reduced to checking recognizability of relations $R_q$
where $q$ is reachable from an asynchronous cycle.

Let us now consider the general case of a $k$-ary deterministic rational relation $R$.
Again, we can ignore the endmarker $\dashv$ since appending $(\dashv,\dots,\dashv)$ to $R$
preserves (non)synchronicity.
Hence, for the rest of this section we assume that $R$ is given by a deterministic $k$-tape automaton $\A = (Q,\Sigma,q_0,\delta,F)$ with $R = R(\A)$.
Moreover, we assume that every state in $Q$ is reachable from $q_0$.
A cycle in $\A$ reading $(v_1,\dots,v_k)$ \emph{induces} a partition $P$ on the components $[1,k]$
where two components $i$ and $j$ are in the same block in $P$ if and only if $|v_i| = |v_j|$.
For a state $q \in Q$ we define the partition $P_q$ as the coarsest refinement of all partitions induced by a cycle from which $q$ is reachable.
As before, let $R_q$ be the relation recognized from state $q$.

\begin{restatable}{lemma}{lempq}\label{lem:Pq}
For every $q \in Q$, $P_q$ is computable in time polynomial in the size of $\A$.
\end{restatable}
\begin{proof}
The algorithm proceeds as follows.
As a first step we compute for each $q \in Q$ the coarsest refinement $S_q$ of all partitions that are induced by \emph{simple} cycles on $q$.
Check for every $1 \le i < j \le k$ if there exists a simple cycle $q \xrightarrow{\bm{v}} q$ such that $|v_i| \ne |v_j|$
and if so, store it as a constraint that $i$ and $j$ are in different blocks.
Then $S_q$ is the coarsest partition of $[1,k]$ that fulfills all stored constraints.
Note that the existence of a simple cycle $q \xrightarrow{\bm{v}} q$ with $|v_i| \ne |v_j|$ can be checked in nondeterministic logspace
by storing the current length difference of the words in the $i$-th and $j$-th component on the guessed path in a counter whose value is bounded by $|Q|$.

We claim that $P_q = S_{q_1} \sqcap \dots \sqcap S_{q_n}$ 
where $q_1,\dots,q_n$ are the states from which $q$ is reachable.
By definition, $P_q$ is finer than $S_{q_1} \sqcap \dots \sqcap S_{q_n}$.
For the other direction let $P$ be a partition induced by a cycle $c$ from which $q$ is reachable.
For the sake of contradiction assume that there exist $1 \le i < j \le k$ that are in different blocks in $P$ but in the same block in $S_{q_\ell}$ for all $\ell \in [1,n]$.
Since any cycle contains a simple cycle, there exists a simple cycle $p \xrightarrow{\bm{v}} p$ that is contained in $c$.
By assumption, it holds that $|v_i| = |v_j|$ which means that after removing $p \xrightarrow{\bm{v}} p$ from $c$, 
the cycle $c$ still induces a partition where $i$ and $j$ are in different blocks.
Furthermore, $q$ is still reachable from $c$.
We can repeat this argument until $c$ is a simple cycle inducing a partition where $i$ and $j$ are in different blocks, a contradiction.
\end{proof}

Recall that, for binary relations we tested recognizability for all states reachable from asynchronous cycles.
For higher arity relations, we need to test whether each relation $R_q$ conforms to $P_q$.

\begin{figure}
\centering
\begin{tikzpicture}[yscale=-1, every node/.style={inner sep=0,outer sep=0}]
\node[draw, fill = black!20, minimum height = 10pt, minimum width = 40pt] at (0,0) {\scriptsize $u_1$};
\node[draw, fill = black!20, minimum height = 10pt, minimum width = 50pt] at (5pt,10pt) {\scriptsize $u_2$};
\node[draw, fill = black!20, minimum height = 10pt, minimum width = 85pt] at (22.5pt,20pt) {\scriptsize $u_3$};
\node[draw, fill = black!20, minimum height = 10pt, minimum width = 85pt] at (22.5pt,30pt) {\scriptsize \qquad\qquad\qquad\qquad $u_4$};
\node[draw, minimum height = 10pt, minimum width = 20pt] at (30pt,0) {\scriptsize $s_{i,1}$};
\node[draw, minimum height = 10pt, minimum width = 20pt] at (40pt,10pt) {\scriptsize $s_{i,2}$};
\node[draw, minimum height = 10pt, minimum width = 105pt] at (92.5pt,0) {\scriptsize $\bot ~ ~ \bot ~ ~ \bot ~ ~ \bot ~ ~ \bot ~ ~ \bot ~ ~ \bot ~ ~ \bot ~ ~ \bot$};
\node[draw, minimum height = 10pt, minimum width = 95pt] at (97.5pt,10pt) {\scriptsize $\bot ~ ~ \bot ~ ~ \bot ~ ~ \bot ~ ~ \bot ~ ~ \bot ~ ~ \bot ~ ~ \bot$};
\node[draw, minimum height = 10pt, minimum width = 40pt] at (125pt,20pt) {\scriptsize $\bot ~ ~ \bot ~ ~ \bot$};
\node[draw, minimum height = 10pt, minimum width = 40pt] at (85pt,20pt) {\scriptsize $t_{i,j,3}$};
\node[draw, fill = black!20, minimum height = 10pt, minimum width = 60pt] at (95pt,30pt) {};
\node[draw, minimum height = 10pt, minimum width = 20pt] at (135pt,30pt) {\scriptsize $t_{i,j,4}$};
\draw[red,opacity=0.7,semithick] (65pt,-10pt) to (65pt,40pt);
\draw [decorate,decoration = {brace}] (65pt,40pt) -- (-20pt,40pt) node[pos=0.5,below=5pt]{\footnotesize $x_i$};
\draw [decorate,decoration = {brace}] (145pt,40pt) -- (65pt,40pt) node[pos=0.5,below=5pt]{\footnotesize $y_{i,j}$};
\end{tikzpicture}
\caption{If there are $h$ pairwise $\approx^{R_q}_{[1,2]}$-inequivalent tuples $\bm{s_1}, \dots, \bm{s_h}$
then also the Myhill-Nerode equivalence $\sim_{\otimes R}$ has at least $h$ classes (witnessed by $x_1, \dots, x_h$).}
\label{fig:completeness}
\end{figure}

\begin{lemma}\label{lem:completeness}
If $R_q$ does not conform to $P_q$ for some state $q \in Q$, then $R$ is not synchronous.
\end{lemma}
\begin{proof}
Assume that $R_q$ does not conform to $P_q$.
By \Cref{thm:partition} there exists a partition $P$ induced by
a cycle $p \xrightarrow{\bm{v}} p$ so that $q$ is reachable from $p$
and $R_q$ does not conform to $P$.
By permuting components, we can assume that $|v_1| \le \dots \le |v_k|$.
Hence $P$ is a partition of $[1,k]$ into intervals $B_1, \dots, B_n$, which are listed in ascending order.
Since the intervals $B_1 \cup \cdots \cup B_i$ for $i \in [1,n]$ generate $P$,
there exists an index $r \in [1,n]$ such that
$\approx_{B_1 \cup \cdots \cup B_r}^{R_q}$ has infinite index by \Cref{lem:conforms-inf-index}.
Set $[1,m] \define = B_1 \cup \cdots \cup B_r$.
Observe that for any number $b \in \N$ there exists a run $q_0 \xrightarrow{\bm{u}} q$
such that $|u_i| + b \le |u_j|$ for all $i \in [1,m]$ and $j \in [m+1,k]$.
Such runs can be constructed by traversing the cycle $p \xrightarrow{\bm{v}} p$ sufficiently often.

We show that for every $h \in \N$, the Myhill-Nerode equivalence relation $\sim_{\otimes R}$
of the language $\otimes R$ of convolutions has at least $h$ classes.
This proves that $\otimes R$ is not regular and therefore $R$ is not synchronous.

Let $h \in \N$.
Since $\approx_{[1,m]}^{R_q}$ has infinite index
there are tuples $\bm{s_i} \define= (s_{i,1},\dots,s_{i,m})$ for $i \in [1,h]$ and
$\bm{t_{i,j}} \define= (t_{i,j,m+1},\dots,t_{i,j,k})$ for $1 \le i < j \le h$
such that $(\bm{s_i},\bm{t_{i,j}}) \in R_q$ if and only if $(\bm{s_j}, \bm{t_{i,j}}) \notin R_q$ for all $1 \le i < j \le h$.

Let $b \define = \max\{ |s_{i,j}| \mid i \in [1,h], \, j \in [1,m]\}$.
By the observation above there exists a run $q_0 \xrightarrow{\bm{u}} q$
such that $|u_i| + b \le |u_j|$ for all $i \in [1,m]$ and $j \in [m+1,k]$.
Therefore, there exists a number~$\ell$ such that all words in the $m$-tuple $(u_1, \dots, u_m) \bm{s_i}$
have length at most $\ell$, and all words in the $(k-m)$-tuple $(u_{m+1}, \dots, u_k)$
have length at least $\ell$, see~\Cref{fig:completeness} for an illustration.
Since $\A$ is deterministic
we have $\bm{u} (\bm{s_i},\bm{t_{i,j}}) \in R$ if and only if $\bm{u} (\bm{s_j},\bm{t_{i,j}}) \notin R$
for all $1 \le i < j \le h$.
This can be turned into a proof that $\otimes R$ has at least $h$ Myhill-Nerode classes.
For a word $w$ of length at least $\ell$,
we denote by $\pre_\ell(w)$ the prefix of $w$ of length $\ell$ and by $\suf_\ell(w)$ the suffix of $w$ after $\pre_\ell(w)$.
For all $i \in [1,h]$ define
\begin{align*}
x_i \define=&\ u_1 s_{i,1} \otimes \dots \otimes u_m s_{i,m} \otimes \pre_\ell (u_{m+1}) \otimes \dots \otimes \pre_\ell (u_k)
\end{align*}
and for all $1 \le i < j \le h$ define
\begin{align*}
y_{i,j} \define=&\ \varepsilon \otimes \dots \otimes \varepsilon \\
&\otimes \suf_\ell (u_{m+1} t_{i,j,m+1}) \otimes \dots \otimes \suf_\ell (u_k t_{i,j,k}).
\end{align*}
Observe that $x_i y_{i,j}$ and $x_j y_{i,j}$ are the convolutions of the tuples
$\bm{u} (\bm{s_i},\bm{t_{i,j}})$ and $\bm{u} (\bm{s_j},\bm{t_{i,j}})$, respectively,
and therefore $x_i \not \sim_{\otimes R} x_j$ for all $1 \le i < j \le h$.
\end{proof}

Testing whether a relation conforms to a partition is an, a priori, more difficult problem than recognizability,
and it is not clear how to decide it for deterministic rational relations.
Instead, we will summarize the components inside each partition block into a single component,
and test recognizability for the summarized relation.

In the following we always assume that the blocks of a partition $P = \{B_1,\dots,B_n\}$ are ordered so that
$\min(B_1) < \dots < \min(B_n)$,
and each block $B_i = \{b_{i,1},\dots,b_{i,|B_i|}\}$ is given such that $b_{i,1} < \dots < b_{i,|B_i|}$.
For a relation $R \subseteq (\Sigma^*)^k$ and the partition $P$ of $[1,k]$ as above we define the \emph{summarized relation}
\begin{multline*}
R^P \define = \{(u_{b_{1,1}} \otimes \dots \otimes u_{b_{1,|B_1|}},\dots,u_{b_{n,1}} \otimes \dots \otimes u_{b_{n,|B_n|}}) \\ \mid (u_1,\dots,u_k) \in R\}.
\end{multline*}
We write $R_q^\otimes$ for $R_q^{P_q}$.
Under the assumption that $R_q^\otimes$ is deterministic rational,
we can test if $R_q$ conforms to $P_q$.

\begin{lemma}\label{lem:rec-conforms}
If $R^P$ is deterministic rational,
then $R^P$ is recognizable if and only if $R$ conforms to $P$.
\end{lemma}
\begin{proof}
Let $P = \{B_1,\dots,B_n\}$.
By \Cref{prop:rat-fin-idx} the summarized relation $R^P$ is recognizable if and only if 
$\approx_{[1,i]}^{R^P}$ has finite index for all $i \in [1,n-1]$.
Let $B_{[1,i]} \define= B_1 \cup \dots \cup B_i$.
By \Cref{lem:conforms-inf-index} we have that $R$ conforms to $P$ if and only if $\approx_{B_{[1,i]}}^R$ has finite index for all $i \in [1,n-1]$.
The claim follows since $\approx_{[1,i]}^{R^P}$ has finite index if and only if $\approx_{B_{[1,i]}}^R$ has finite index.
\end{proof}

We are now ready to give the reduction from synchronicity to recognizability proving \Cref{thm:drat-sync}.
For a partition $P$ of $[1,k]$ we write $Q_P \define= \{q \in Q \mid P_q = P\}$.
We partition $Q$ into layers $L_1, \dots, L_k$ where $L_t = \{ q \in Q \mid |P_q| = t \}$.
Observe that all states reachable from a state $q \in L_t$ are contained in $L_t \cup \cdots \cup L_k$.
The algorithm processes the layers $L_k, L_{k-1}, \dots, L_1$ in descending order.
For each state $q$ in layer $L_t$ the algorithm
(i) constructs a deterministic multitape automaton $\A^\otimes_q$ for $R_q^\otimes$,
(ii) tests whether $R_q^\otimes$ is recognizable,
(iii) and, if so, constructs an independent multitape automaton $\I^\otimes_q$ for $R_q^\otimes$.
For layer $L_k$ the automaton $\A^\otimes_q$ is simply the automaton $\A$ with initial state $q$.
For the other layers the automaton $\A^\otimes_q$ will be built from the automata $\I^\otimes_{q'}$
from the previous layers, which will be explained below (\cref{lem:construction}).
The automata $\I^\otimes_q$ are constructed from $\A^\otimes_q$ using \Cref{thm:twoexp}, which is correct under the assumption
that $R_q^\otimes$ is indeed recognizable.
If one of the recognizability tests is negative, the algorithm terminates and reports that $R$ is not synchronous.
Otherwise, if all recognizability tests succeed, the algorithm reports that $R$ is synchronous.

Let us argue that the algorithm is correct.
If one of the relations $R_q^\otimes$ is deterministic rational but not recognizable,
then $R$ is indeed not synchronous by \Cref{lem:rec-conforms,lem:completeness}.
If all recognizability tests succeed, then in particular the relation $R_{q_0}^\otimes$ is recognizable.
This easily implies synchronicity of $R$.
In fact, we can make the algorithm slightly more efficient.
Observe that the recognizability tests in layer $L_1$ will always be positive
since the relations $R_q^\otimes$ in layer $L_1$ are regular languages.
Therefore, we can skip processing the last layer $L_1$ and also skip constructing the independent multitape automata in layer $L_2$.

It remains to show how to construct the automaton $\A_q^\otimes$ witnessing that $R_q^\otimes$ is deterministic rational.
\begin{restatable}{lemma}{lemconstruction}\label{lem:construction}
Given $q \in L_t$ and independent multitape automata $\I_{q'}^\otimes$ for $R_{q'}^\otimes$ for all $q' \in L_{t+1} \cup \cdots \cup L_k$,
one can compute $\A_q^\otimes$ in time polynomial in the sizes of the $\I_{q'}^\otimes$ and exponential in the size of $\A$.
\end{restatable}
The detailed construction can be found in \Cref{sec:app-sync-rec}.
The idea is that $\A_q^\otimes$ consists of two parts.
Let $q \in Q_P$ for a partition $P = \{B_1,\dots,B_n\}$.
In the first part $\A_q^\otimes$ simulates $\A$ over the states in $Q_P$.
In that part it is possible for $\A_q^\otimes$ to read the components within a class of $P$ synchronously
since by definition of $Q_P$ the difference of the tape positions between those components is bounded by $|Q|$.
Thus, it suffices for $\A_q^\otimes$ to store a word (read like a queue) of length at most $|Q|$ for each component
and simulate $\A$ either on the next symbol in the queue or on the current input symbol if the queue is empty.
If it simulates $\A$ on the input symbol of the corresponding tape, we store the symbols read in the other components in the queues of that components.
The simulation of $\A$ on the queue is handled like an $\varepsilon$-transition where nothing is read from the input.
Those $\varepsilon$-transitions can be removed and do not lead to nondeterminism since there is no branching of $\varepsilon$-transitions possible.
If the simulation of $\A$ leads to a state $q'$ that is not contained in $Q_P$, 
i.e., $q' \in Q_{P'}$ for some partition $P' = \{B'_1,\dots,B'_{n'}\}$ that is strictly finer than $P$,
then $\A_q^\otimes$ changes to the second part.
In the second part $\A_q^\otimes$ simulates the independent $n'$-tape automaton $\I_{q'}^\otimes = (\A_1,\dots,\A_{n'},F_{q'})$.
For each class $B_i$, for $i = 1,\dots,n$, the DFAs $\A_j$ that are responsible for components contained in $B_i$ are simulated in parallel
either on the queue or the current input symbol.
After the whole input was read, $\A_q^\otimes$ checks with final states whether the remaining content of the queues leads in each $\A_i$ to some state $f_i$
such that $(f_1,\dots,f_{n'}) \in F_{q'}$.

Finally, we argue that the running time of the algorithm is $2(k-2)$-fold exponential in the automaton size $|\A|$.
Inductively, we prove that in layer $L_t$ we can construct the automata $\A_q^\otimes$ in $2(k-t)$-fold exponential time
and the automata $\I_q^\otimes$ in $2(k-t+1)$-fold exponential time.
In particular, the automata sizes are bounded by their respective construction times.
In layer $L_k$ the automata $\A_q^\otimes$ are constructed in polynomial time.
In the other layers $L_t$ the automata $\A_q^\otimes$ are constructed by \cref{lem:construction}
in $2(k-t)$-fold exponential time.
Each automaton $\I_q^\otimes$ is constructed in double exponential time in the size of $\A_q^\otimes$ (\Cref{thm:twoexp}),
which is $2(k-t+1)$-fold exponential in $|\A|$.
Furthermore, each recognizability test on $\A_q^\otimes$ in layer $L_t$ where $t \in [3,k]$
takes double exponential time in $|\A_q^\otimes|$ by \cref{thm:drat-recognizability},
which is $2(k-t+1)$-fold exponential in $|\A|$.
The relations $R_q^\otimes$ in layer $L_2$ are binary and therefore recognizability can be tested
in polynomial time in $|\A_q^\otimes|$, which is $2(k-2)$-fold exponential in $|\A|$.


\section*{Acknowledgments}
The authors thank Stefan Göller, Anthony W. Lin, and Georg Zetzsche for helpful discussions.

\begin{wrapfigure}[4]{r}{0.05\textwidth}
\vspace{-15pt}
  \begin{center}
    \includegraphics[width=0.05\textwidth]{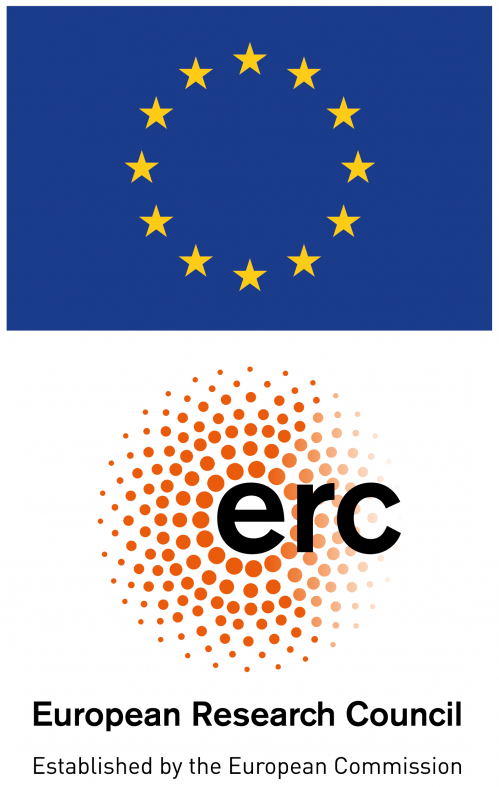}
  \end{center}
\end{wrapfigure}
This work is funded by the European Union (ERC, AV-SMP, 759969 and ERC, FINABIS, 101077902).
Views and opinions expressed are however those of the author(s) only and do not
necessarily reflect those of the European Union or the European Research Council Executive Agency.
Neither the European Union nor the granting authority can be held responsible for them.



\bibliographystyle{IEEEtran}
\bibliography{IEEEabrv,bib}

\appendix




\subsection{Discussion of Proof of Lemma~3.5 in \cite{CartonCG06}}
\label{sec:app-ccg}

At the end of the proof of Lemma~3.5 in \cite{CartonCG06} we have the following setting:
There exists a cycle $q \xrightarrow{(u_2,\bm{t}_2)} q$ where $u_2$ is null-transparent,
and either $u_2u_3vw \not \approx_1^{R_q} u_3vw$ or $u_2u_3w \not \approx_1^{R_q} u_3w$.
Observe this precisely matches the situation in point \ref{it:sep-one}) of \cref{lem:carton}.
Let us assume that $u_2u_3vw \not \approx_1^{R_q} u_3vw$ holds,
i.e.\ there exists $\bm{z} \in (\Sigma^*)^{k-1}$ such that
$R_q$ separates $(u_2u_3vw,\bm{z})$ and $(u_3vw,\bm{z})$.
First we have $u_2^{i+1} u_3vw \not \approx_1^{R_q} u_2^i u_3 vw$ for all $i \ge 0$
since $R_q$ separates $(u_2^{i+1} u_3vw,\bm{t}_2^i \bm{z})$ and $(u_2^i u_3vw,\bm{t}_2^i \bm{z})$.
Then, it is claimed that $u_2^{i+j} u_3vw \not \approx_1^{R_q} u_2^i u_3 vw$ for all $i \ge 0$, $j > 0$.
Towards a contradiction assume that $u_2^{i+K} u_3vw \approx_1^{R_q} u_2^i u_3 vw$ for some $i \ge 0$, $K > 0$.
Then, it is deduced that $u_2^{i+\lambda K} u_3vw \approx_1^{R_q} u_2^i u_3 vw$
and $u_2^{i+1+\lambda K} u_3vw \approx_1^{R_q} u_2^{i+1} u_3 vw$ for all $\lambda \ge 0$.
The authors of \cite{CartonCG06} seem to assume that $\approx_1$ is a \emph{left-congruence},
i.e.\ $x \approx_1 y$ implies $ax \approx_1 ay$, which is not true in general.

As a simple fix, one can replace $\approx_1^{R_q}$ by the following left-congruence
\[
	x~\dot \approx_1^{R_q} y \define \iff \text{for all } z \in \Sigma^* \colon zx \approx_1^{R_q} zy.
\]
With this replacement, the verbatim proof in \cite{CartonCG06} shows that $\dot \approx_1^{R_q}$ has infinite index.
It is not difficult to show that this contradicts the assumption made in the lemma that $\approx_1^R$ has finite-index.


\subsection{Proofs for \Cref{sec:rec-sync}}
\label{sec:app-rec-sync}

\recsync*

\begin{proof}

Let $R \subseteq (\Sigma^*)^k$ be a rational relation and $\Sigma' \define= \Sigma \cup \{\#,\vdash\}$ with $\#,\vdash \notin \Sigma$.
We show that $R$ is recognizable if and only if 
\[
R' \define= \{(\#^{n_1}\vdash w_1,\dots,\#^{n_k}\vdash w_k) \mid n_i \ge 0 \text{ and } \bm{w} \in R\}
\]
is synchronous.
Since from a $k$-tape automaton for $R$ we can easily construct a $k$-tape automaton for $R'$ in logspace
while preserving determinism,
this shows that recognizability is logspace reducible to synchronicity for both relations in $\Rat$ and relations in $\DRat$.

For the ``only if'' direction assume that $R = \bigcup_{i=1}^n L_{i,1} \times \dots \times L_{i,k}$ for regular languages $L_{i,j} \subseteq \Sigma^*$.
Then $R' = \bigcup_{i=1}^n \{\#\}^* \{\vdash\} L_{i,1} \times \dots \times \{\#\}^* \{\vdash\} L_{i,k}$
where the $\{\#\}^* {\{\vdash\}} L_{i,j} \subseteq (\Sigma')^*$ are clearly regular languages.
Thus, $R'$ is recognizable and therefore also synchronous.

For the ``if'' direction we assume that $R$ is not recognizable which by \Cref{prop:rat-fin-idx} means that
$\approx_r$ has infinite index for some $r \in [1,k-1]$.
Thus, there are words $v_1, v_2, \dots \in \Sigma^*$ and tuples $\bm{w_{i,j}} = (w_{i,j,1},\dots,w_{i,j,k-1}) \in (\Sigma^*)^{k-1}$
such that $v_i \odot_{r} \bm{w_{i,j}} \in R$ if and only if $v_j \odot_{r} \bm{w_{i,j}} \notin R$ for all $i < j$.
We show that the Myhill-Nerode equivalence relation $\sim_{\otimes R'}$ of the language of convolutions $\otimes R'$ has infinite index.
This implies that $\otimes R'$ cannot be regular and therefore $R'$ is not synchronous.
Let 
\[x_i \define= \#^{n_i} \vdash \otimes \dots \otimes \#^{n_i} \vdash \otimes \vdash v_i \otimes \#^{n_i} \vdash \otimes \dots \otimes \#^{n_i} \vdash\]
with $\vdash v_i$ in the $r$-th component and $n_i \define= |v_i|$ for all $i \ge 1$
and let
\[y_{i,j} \define= w_{i,j,1} \otimes \dots \otimes w_{i,j,r-1} \otimes \varepsilon \otimes w_{i,j,r} \otimes \dots \otimes w_{i,j,k-1}\]
for all $i < j$.
Then we have that $x_i y_{i,j} \in \otimes R'$ if and only if $x_j y_{i,j} \notin \otimes R'$
which means that $x_i \not \sim_{\otimes R'} x_j$ for all $i < j$.

\end{proof}


\label{sec:app-sync-rec}

\lemconstruction*

\begin{proof}

Let $R$ be given by a deterministic $k$-tape automaton $\A = (Q,\Sigma,q_0,\delta,F)$ with partition of states $Q = Q_1 \uplus \dots \uplus Q_k$.
Let $q \in Q_P$ for a partition $P = \{B_1,\dots,B_t\}$ of $\{1,\dots,k\}$.
We show how to construct a deterministic $t$-tape automaton 
$\A_q^\otimes = (Q',\Sigma',q'_0,\delta',F')$ with $Q' = Q'_1 \uplus \dots \uplus Q'_t$
and $\Sigma' \define= \Sigma_\bot^{|B_1|} \cup \dots \cup \Sigma_\bot^{|B_t|} \cup {\{\dashv\}}$
where $\Sigma_\bot \define= \Sigma \cup \{\bot\}$ and $\bot,\dashv \notin \Sigma$
that recognizes the relation $\{\bm{w}(\dashv,\dots,\dashv) \mid \bm{w} \in R_q^\otimes\}$
assuming that we already constructed independent multitape automata $\I_{q'}^\otimes$ for $q' \in L_{t+1} \cup \cdots \cup L_k$.
The automaton $\A_q^\otimes$ consists of two parts.
The first part simulates $\A$ over the states in $Q_P$ and the second part simulates the independent multitape automata of states of higher layers than $q$.

\paragraph*{First part}
The first part of $\A_q^\otimes$ consists of states of the form $(p,\bm{w})$
with $p \in Q_P$ and $w_i \in \Sigma_\bot^{\le |Q|}$ for $1 \le i \le k$.
Intuitively, $p$ stores the state of $\A$ the simulation is currently at and
$\bm{w}$ stores a queue for every component of $\A$ of symbols that $\A$ still has to be simulated on.
The initial state is $q'_0 \define= (q,\varepsilon,\dots,\varepsilon)$.
For input $a' = (a_1,\dots,a_{|B_i|}) \in \Sigma_\bot^{|B_i|}$
for $i \in [1,t]$ and $a_j \in \Sigma$ for $j \in [1,|B_i|]$ and
state $(p,\bm{w})$ with $p \in Q_{b_{i,j}}$, $w_{b_{i,j}} = \varepsilon$, and 
$p' \define= \delta(p,a_j) \in Q_P$ we let $(p,\bm{w}) \in Q'_i$ and
$\delta'((p,\bm{w}),a') \define= (p',\bm{w'})$ where
\[
w'_{b_{i',j'}} \define= 
\begin{cases}
w_{b_{i,j'}} a_{j'}, &\text{if } i' = i \text{ and } j' \ne j \\
w_{b_{i',j'}}, &\text{otherwise}
\end{cases}
\]
for all $i' \in [1,t]$ and $j' \in [1,|B_{i'}|]$.
For state $(p,\bm{w})$ with $p \in Q_i$ for some $i \in [1,k]$, $w_i = au$ for some $a \in \Sigma$ and $u \in \Sigma_\bot^*$, and 
$p' \define= \delta(p,a) \in Q_P$ we define
$\delta'((p,\bm{w}),\varepsilon) \define= (p',w_1,\dots,w_{i-1},u,w_{i+1},\dots,w_k)$.
Note that these $\varepsilon$-transitions can be eliminated without introducing nondeterminism since there is no branching possible.
On each state of the form $(p,\varepsilon,\dots,\varepsilon)$ with $p \in F$ we append a chain of transitions reading the endmarker $\dashv$ in every component
and mark the last state of that chain (which is a sink state) as final.

\paragraph*{First to second part}
We now define the transitions from states of the first part to states of the second part.
Let $p' \in Q_{P'}$ for partition $P' = \{B'_1,\dots,B'_{t'}\}$ that is strictly finer than $P$ and
$\I_{p'}^\otimes = (\A_1,\dots,\A_{t'},F_{p'})$ be an independent $t'$-tape automaton for $R_{p'}^\otimes$ with $\A_i = (Q^i,\Sigma_\bot^{|B'_i|},\delta^i,q_0^i)$.
The second part consists of states of the form $(q^1,\dots,q^{t'},\bm{w},\bm{e},m)$ with
$q^i \in Q^i$ for $1 \le i \le t'$, $\bm{w}$ is as in the first part, $e_i \in \{0,1\}$ for $1 \le i \le t'$, and $1 \le m \le t+1$.
Intuitively, $q^i$ is the state the simulation of $\A_i$ is currently at, $e_i$ stores in a bit whether the simulation of $\A_i$ is finished, and 
$m$ indicates which component of $\A_q^\otimes$ is currently read.
For input $a' = (a_1,\dots,a_{|B_i|}) \in \Sigma_\bot^{|B_i|}$
for $i \in [1,t]$ and $a_j \in \Sigma$ for $j \in [1,|B_i|]$ and
state $(p,\bm{w})$ with $p \in Q_{b_{i,j}}$, $w_{b_{i,j}} = \varepsilon$, and 
$\delta(p,a_j) = p'$ we let $(p,\bm{w}) \in Q'_i$ and
$\delta'((p,\bm{w}),a') \define= (q_0^1,\dots,q_0^{t'},\bm{w'},0,\dots,0,1)$ where
\[
w'_{b_{i',j'}} \define= 
\begin{cases}
w_{b_{i,j'}} a_{j'}, &\text{if } i' = i \text{ and } j' \ne j \\
w_{b_{i',j'}}, &\text{otherwise}
\end{cases}
\]
for all $i' \in [1,t]$ and $j' \in [1,|B_{i'}|]$.
For state $(p,\bm{w})$ with $p \in Q_i$ for some $i \in [1,k]$, $w_i = au$ for some $a \in \Sigma$ and $u \in \Sigma_\bot^*$, and 
$\delta(p,a) = p'$ we define
$\delta'((p,\bm{w}),\varepsilon) \define= (q_0^1,\dots,q_0^{t'},w_1,\dots,w_{i-1},u,w_{i+1},\dots,w_k,0,\dots,0,1)$.
Note that these $\varepsilon$-transitions can be eliminated again.

\paragraph*{Second part}
Let $f \colon \{1,\dots,t'\} \to \{1,\dots,t\}$ with $i' \mapsto i$ such that $B'_{i'} \subseteq B_i$.
Note that $f$ is well defined since $P'$ is finer that $P$.
For input $a' = (a_{b_{i,1}},\dots,a_{b_{i,|B_i|}}) \in \Sigma_\bot^{|B_i|} \setminus \{(\bot,\dots,\bot)\}$ for $i \in [1,t]$ and
state $(q^1,\dots,q^{t'},\bm{w},\bm{e},i)$ we let $(q^1,\dots,q^{t'},\bm{w},\bm{e},i) \in Q'_i$ and
$\delta'((q^1,\dots,q^{t'},\bm{w},\bm{e},i),a') \define= (p^1,\dots,p^{t'},\bm{w'},\bm{e'},i)$ where for all $i' \in [1,t']$,
\[p^{i'} \define= \delta^{i'}(q^{i'},(c_{b'_{i',1}},\dots,c_{b'_{i',|B'_{i'}|}}))\] 
if $f(i') = i$, $c_{b'_{i',j'}} \ne \bot$ for some $j'$, and $e_{i'} = 0$,
\[p^{i'} \define= q^{i'}\]
if $f(i') \ne i$ or $f(i') = i$ and $c_{b'_{i',j'}} = \bot$ for all $j'$, and 
otherwise $p^{i'}$ is undefined.
Here, we set
\[
c_{b'_{i',j'}} \define= 
\begin{cases}
a_{b'_{i',j'}}, &\text{if } w_{b'_{i',j'}} = \varepsilon \\
c, &\text{if } w_{b'_{i',j'}} = cu \text{ for some } c \in \Sigma_\bot, u \in \Sigma_\bot^*
\end{cases}
\]
for all $i' \in [1,t']$ with $f(i') = i$ and $j' \in [1,|B'_{i'}|]$.
Furthermore, we define
\[
w'_{b'_{i',j'}} \define= 
\begin{cases}
u a_{b'_{i',j'}}, &\text{if } f(i') = i \text{ and } w_{b'_{i',j'}} = c u\\
w_{b'_{i',j'}}, &\text{otherwise}
\end{cases}
\]
and
\[
e'_{i'} \define= 
\begin{cases}
1, &\text{if } f(i') = i \text{ and } c_{b'_{i',j'}} = \bot \text{ for all } j'\\
e_{i'}, &\text{otherwise}
\end{cases}
\]
for all $i' \in [1,t']$ and $j' \in [1,|B'_{i'}|]$.
For $i \in [1,t]$ and
state $(q^1,\dots,q^{t'},\bm{w},\bm{e},i)$ we let $(q^1,\dots,q^{t'},\bm{w},\bm{e},i) \in Q_i$ and
$\delta'((q^1,\dots,q^{t'},\bm{w},\bm{e},i),\dashv) \define= (q^1,\dots,q^{t'},\bm{w},\bm{e},i+1)$.
A state of the form $(q^1,\dots,q^{t'},\bm{w},\bm{e},n+1)$ with $w_i = u_i v_i$ for $u_i \in \Sigma^*$ and $v_i \in \{\bot\}^*$
is final if for all $i' \in [1,t']$ we have that $\A_{i'}$ reaches state $f^{i'}$ from $q^{i'}$ on reading
$u_{b'_{i',1}} \otimes \dots \otimes u_{b'_{i',|B'_{i'}|}}$ and $(f^1,\dots,f^{t'}) \in F_{p'}$.
\end{proof}

\end{document}